\documentclass[oneside,11pt]{article}
\usepackage[utf8]{inputenc}
\usepackage{amsmath,amsfonts,amssymb,amsthm}
\usepackage{dsfont}
\usepackage[margin=1in]{geometry} 
\usepackage{graphicx}
\usepackage{float}
\usepackage{tikz-cd}
\usepackage{graphicx}
\usepackage[font=small,labelfont=bf]{caption}
\usepackage[colorinlistoftodos]{todonotes}
\usepackage{wrapfig}

\usepackage{array}

\usepackage{hyperref}
\hypersetup{
    colorlinks,
    citecolor=black,
    filecolor=black,
    linkcolor=black,
    urlcolor=black
}

\usepackage{enumitem}

\newcommand{\R}{\mathbb R}
\newcommand{\C}{\mathbb C}
\renewcommand{\H}{\mathbb H}
\newcommand{\D}{\mathbb D}

\newcommand{\Tr}{\text{Tr}}
\newcommand{\Vol}{\text{Vol}}
\renewcommand{\P}{\mathbb P}

\renewcommand{\Re}{\text{Re}}
\renewcommand{\Im}{\text{Im}}
\newcommand{\dom}{\text{dom}}

\newcommand{\vare}{\varepsilon}

\DeclareMathOperator{\zdet}{det_\zeta}
\DeclareMathOperator{\dist}{dist}
\newtheorem{theorem}{Theorem}

\newtheorem{lemma}{Lemma}

\theoremstyle{definition}

\newtheorem{definition}{Definition}
\theoremstyle{remark}
\newtheorem{remark}{Remark}

\usepackage{lipsum}

\title{Polyakov-Alvarez Formula for \\Curvilinear Polygonal Domains with Slits}
\author{Ellen Krusell\\ \textit{KTH Royal Institute of Technology}\\ \texttt{ekrusell@kth.se}}
\date{}
\begin{document}
\maketitle
\pagenumbering{arabic}
\abstract{We consider the $\zeta$-regularized determinant of the Friedrichs extension of the Dirichlet Laplace-Beltrami operator on curvilinear polygonal domains with corners of arbitrary positive angles. In particular, this includes slit domains. We obtain a short time asymptotic expansion of the heat trace using a classical patchwork method. This allows us to define the $\zeta$-regularized determinant of the Laplacian and prove a comparison formula of Polyakov-Alvarez type for a smooth and conformal change of metric.}
\section{Introduction}
The \emph{Dirichlet Laplace-Beltrami} operator $\Delta_{(M,g)}$ on a compact manifold $(M,g)$ with smooth (non-empty) boundary has, with our sign convention, a discrete and positive spectrum
$$0<\lambda_1\leq \lambda_2\leq \lambda_3\leq ....$$
with $\lambda_n\to\infty$ as $n\to\infty$. Hence, it is not possible to define $\det \Delta_{(M,g)}$ in the classical sense. One can, however, define a regularized version of the determinant via the following procedure: Define the \emph{spectral $\zeta$-function} by
$$\zeta(s)=\sum_{n\geq 1}\lambda_n^{-s},\quad \Re s>1,$$
where the right-hand side converges by Weyl's law~\cite{Weyl}. As observed in~\cite{RaySinger}, it turns out that the right-hand side can be expressed as an integral involving $\Tr(e^{-t\Delta_{(M,g)}})$. By employing a short time asymptotic expansion of $\Tr(e^{-t\Delta_{(M,g)}})$ one can then show that $\zeta(s)$ may be analytically continued to a neighborhood of $0$. The formal computation 
$$\zeta'_{(M,g)}(s)=\sum_{n\geq 1}-\log\lambda_n \lambda_n^{-s},\ \Re s>1 \leadsto \text{``}\zeta_{(M,g)}'(0)=-\log\prod_{n\geq 1}\lambda_n,\text{''}$$ 
then justifies defining the $\zeta$-\emph{regularized determinant} of $\Delta_{(M,g)}$ by
\begin{equation}\zdet \Delta_{(M,g)} := e^{-\zeta_{(M,g)}'(0)}.\label{eq:zdet}\end{equation}\par
In~\cite{Polyakov}, Polyakov gave a formula for the variation of the $\zeta$-regularized determinant of the Laplace-Beltrami operator on a closed manifold under a smooth conformal change of metric. A similar formula, for compact manifolds with (smooth) boundary, was given by Alvarez in~\cite{Alvarez}. See also~\cite{OPS}. More recently, Polyakov-Alvarez type formulas have been proved in settings of less regularity, for instance in curvilinear polygonal domains where the interior angles belong to the open set $(0,2\pi)$~\cite{AKR}. In the present article, we prove a Polyakov-Alvarez formula for curvilinear polygonal domains where the interior angles belong to $(0,\infty)$. In particular, this includes smooth slit domains, that is, a smooth domain minus a smooth curve intersecting the boundary non-tangentially (where we consider the two sides of the cut out curve as different parts of the boundary, i.e., in the sense of prime ends).
\par 
\subsection{Main results}
Slightly imprecisely (we give the precise definition in Section~\ref{section:cpd}), a curvilinear polygonal domain $(M,g,(p_j),(\alpha_j))$ is a compact surface $M=M^\circ\cup \partial M$ with boundary $\partial M\neq \varnothing$ along with finitely many points $p_1,...,p_n\in\partial M$ and a smooth Riemannian metric $g$ on $M\setminus\{p_1,...,p_n\}$ such that there, for each $j=1,...,n$, exists a smooth isothermal coordinate in a neighborhood of $p_j$ in which $\partial M$ (locally) consists of two smooth boundary arcs forming an interior angle $\alpha_j\pi>0$ at $p_j$.\par 
The Dirichlet heat kernel $H_{(M,g)}$ is the (minimal) fundamental solution to the Dirichlet heat equation on $(M,g)$. The trace of the heat kernel, the heat trace, satisfies
$$\Tr(e^{-t\Delta_{(M,g)}})=\int_{M}H_{(M,g)}(t;p,p)d\Vol_{g}(p).$$
In Section~\ref{section:heattrace} we prove the following short time asymptotic expansion of the heat trace in a curvilinear polygonal domain.
\begin{theorem}\label{thm:heattrace}
Let $(M,g_0,(p_j),(\alpha_j))$ be a curvilinear polygonal domain, $\sigma,\psi\in C^\infty(M,g,(p_j),(\alpha_j))$, and define $g_u:=e^{2u\sigma}g_0$. Then, for each $q\in(0,1/2)$ and each $u\in \R$,
\begin{align}\begin{split}\int_{M}\psi(p) H_{(M,g_u)}(t;p,p)d\Vol_{g_u} = & \frac{1}{4\pi t}\int_M\psi\bigg(1+\frac{t}{3}K_{g_u}\bigg)d\Vol_{g_u}-\frac{1}{8\sqrt{\pi t}}\int_{\partial M}\psi d\ell_{g_u}\\&+\frac{1}{12\pi}\int_{\partial M}\psi k_{g_u}d\ell_{g_u}+\frac{1}{8\pi}\int_{\partial M}\partial_{n_{g_u}}\psi d\ell_{g_u}\\&+\frac{1}{24}\sum_{j=1}^n\psi(p_j)\frac{1-\alpha_j^2}{\alpha_j}+O(t^q),
\end{split}\label{eq:heattracepsi}
\end{align}
where the error is locally uniform in $u$ as $t\to 0+$.
\end{theorem}
Above, $\Vol_g$ and $\ell_g$ denote the area and arc-length measures with respect to $g$ respectively, $K_g$ is the Gaussian curvature, $k_g$ the geodesic curvature, and $\partial_{n_g}$ the outer unit normal derivative.
\begin{remark}In the case where the interior angles satisfy $\alpha_j\pi\in(0,2\pi)$, this agrees with~\cite[Theorem 1.2]{NRS} and~\cite[Corollary 2.3]{AKR} (along with the computations of the coefficients in~\cite[Section 2.2]{AKR}), however in~\cite{NRS,AKR} the error is of order $t^{1/2}\log t$. Possibly one could, using the microlocal techniques of~\cite{NRS}, improve upon the error of Theorem~\ref{thm:heattrace}, but for the purpose of showing the Polyakov-Alvarez formula, the error $O(t^q)$ will suffice.
\end{remark}
It follows from Theorem~\ref{thm:heattrace} that the spectral $\zeta$-function corresponding to a curvilinear polygonal domain can be analytically continued to a neighborhood of $s=0$. See Section \ref{section:zdet} for details. Hence, the $\zeta$-regularized determinant can be defined, for such surfaces, by~\eqref{eq:zdet}. In Section~\ref{section:polalv} we prove the following Polyakov-Alvarez type formula using Theorem~\ref{thm:heattrace}.
\begin{theorem}\label{thm:polyakov}
Let $(M,g_0,(p_j),(\alpha_j))$ be a curvilinear polygonal domain, $\sigma\in C^\infty(M,g_0,(p_j),(\alpha_j))$, and define $g_u=e^{2 u\sigma}g_0$, for $u\in\R$. Then,
\begin{align}\begin{split}\partial_u\log\zdet\Delta_{(M,g_u)} =& -\frac{1}{6\pi}\int_M \sigma K_{g_u}d\Vol_{g_u}-\frac{1}{6\pi}\int_{\partial M}\sigma k_{g_u}d\ell_{g_u}\\ &-\frac{1}{4\pi}\partial_{n_{g_u}}\sigma d\ell_{g_u}-\frac{1}{12}\sum_{j=1}^n\frac{1-\alpha_j^2}{\alpha_j}\sigma(p_j).\end{split}\label{eq:polyakovdifferentiated}\end{align}
In particular,
\begin{align}\begin{split}\log\zdet\Delta_{(M,g_0)}-\log\zdet\Delta_{(M,g_1)}=&\frac{1}{12\pi}\int_M|\nabla_{g_0}\sigma|^2d\Vol_{g_0}+\frac{1}{6\pi}\int_M \sigma K_{g_0} d\Vol_{g_0}\\ &+\frac{1}{6\pi}\int_{\partial M} \sigma k_{g_0}d\ell_{g_0}+\frac{1}{4\pi}\int_{\partial M}\partial_{n_{g_0}}\sigma d\ell_{g_0}\\&+\frac{1}{12}\sum_{j=1}^n\frac{1-\alpha_j^2}{\alpha_j}\sigma(p_j).\end{split}\label{eq:polyakovintegrated}\end{align}
\end{theorem}
\subsection{Motivation} 
The Gaussian free field (GFF) with Dirichlet boundary condition is the Gaussian random distribution whose correlation function is the Green's function for the Dirichlet Laplacian on a region of the complex plane. In recent years, links between the $\zeta$-regularized determinant of the Laplacian and objects from random conformal geometry, in particular those closely related to the GFF, have been established. The heuristic reason for this is that $\zdet\Delta$ may be interpreted as the partition function of the GFF, see, e.g.,~\cite{D09}. (On the discrete level one may see this is by noticing that the GFF has a discrete variant which has the partition function $(\det\Delta_{\text{graph}})^{-1/2}$, where $\Delta_{\text{graph}}$ is the graph Laplacian.) In~\cite{APPS}, connections between Liouville quantum gravity surfaces (which are closely related to the GFF), Brownian loop measure, and $\zeta$-regularized determinants of Laplacians are investigated. See also~\cite{PPS}.\par
Another object from random conformal geometry related to $\zeta$-regularized determinants of Laplacians is the Loewner energy. The Loewner energy is a functional on (deterministic) chords (that is, a simple curve connecting two distinct boundary points) in a simply connected domain. The Loewner energy was first introduced as a large deviation rate function for the family of random curves SLE$_\kappa$ as $\kappa\to 0+$, but has also been found to have strong ties to Teichmüller theory~\cite{W19a,W19b}. The Loewner energy $I(\gamma)$ of a smooth chord $\gamma$ from $-1$ to $1$ in the unit disk $\D=\{z\in\C:|z|<1\}$ relates to determinants of Laplacians through the Loewner potential
$$\mathcal H(\gamma)=\log\zdet \Delta_{(\D,dz)}-\log\zdet\Delta_{(D_1,dz)}-\log\zdet\Delta_{(D_2,dz)},$$
where $D_1$ and $D_2$ are the connected components of $\D\setminus\gamma$ and $dz$ denotes the Euclidean metric, by
\begin{equation}I(\gamma)=12(\mathcal H(\gamma)-\mathcal H([-1,1]))\label{eq:loewner}\end{equation}
provided that $\gamma$ meets $\partial \D$ orthogonally at $-1$ and $1$~\cite{PW}. In a companion paper, we study a variant of the Loewner energy for slits and prove, using Theorem~\ref{thm:polyakov}, a formula analogous to~\eqref{eq:loewner}~\cite{Krusell}. This is an example of a natural occurrence of slit domains, where $\zeta$-regularized determinants can be considered, and serves as motivation for the present article. 
\subsection*{Acknowledgments}
I would like to thank Julie Rowlett for discussions about the project and her related work. I also thank Fredrik Viklund for suggesting the project, many discussions, and for his comments on the draft. Finally, I would like to thank my colleague Andrés Franco Grisales for helpful conversations about differential geometry. This work was supported by a grant from the Knut and Alice Wallenberg foundation. 
\section{Preliminaries}
\subsection{Curvilinear polygonal domains}\label{section:cpd}
Recall that a smooth surface with boundary $M=M^\circ\cup\partial M$ is a topological space which is Hausdorff, second countable, and locally homeomorphic to the closed upper-half plane $\overline\H=\{x=(x_1,x_2)\in\R^2:x_1\geq 0\}$. By the latter we mean that there, for every $p\in M$, is a chart $(U,\varphi)$ consisting of an open neighborhood of $U\ni p$ and a homeomorphism $\varphi:U\to V$ where $V\subset\overline\H$ is open (if $p\in M^\circ$ then we will have $\Im \varphi(p)>0$ and if $p\in\partial M$ then $\Im \varphi(p)=0$). We call a family of charts $\mathcal A=\{(U_\beta,\varphi_\beta)\}_{\beta}$ an atlas if $\cup U_\beta= M$. A smooth structure on $M$ is an atlas $\mathcal A$ of smoothly compatible charts. That is, if $(U_1,\varphi_1),(U_2,\varphi_2)\in \mathcal A$, then $\varphi_2\circ \varphi_1^{-1}\in C^\infty(\varphi_1(U_1\cap U_2))$.\par 
We let $M_0$ denote the topological space $\{(r,\theta):r>0,\theta\in\R\}\subset \R^2$. We endow $M_0$ with (a smooth structure and) a Riemannian metric by declaring that 
$$\pi:M_0\to \C\setminus\{0\}\cong \R^2\setminus\{(0,0)\}$$
$$(r,\theta)\mapsto re^{i\theta}$$ is a local isometry and we refer to this metric on $M_0$ as the Euclidean metric. We will denote the Euclidean metric on $\R^2,\ \C$ by $dx,\ dz$ etc. In order to simplify notation we will often treat $M_0$ as $\C$ (or $\R^2$) without explicitly using the map $\pi$. E.g., we will slightly abuse notation and write $dz$ for $(\pi)^\ast dz$ on $M_0$. 
We let $\hat M_0 = M_0\cup\{0\}$, be the topological space obtained by declaring
$$B_{r}:=B_{\hat M_{0}}(0,r):=\{0\}\cup\{(\rho,\theta):\rho< r\}$$
to be open for all $r>0$. The map $\pi$ is extended to $0$ by $\pi(0)=0$. 
We will also use the notation $|(r,\theta)|:=r$ and $\arg(r,\theta):=\theta$ for $(r,\theta)\in M_0$.
\begin{definition}\label{def:cpd}Let $M$ be a compact surface with boundary $\partial M\neq \varnothing$, with finitely many (distinct) marked points $p_1,...,p_n\in\partial M$, a smooth structure and a smooth Riemannian metric $g$ on $M\setminus\{p_1,...,p_n\}$. We say that $(M,g,(p_j),(\alpha_j))$, where $\alpha_1,...,\alpha_n>0$, is a curvilinear polygonal domain, if there exists, for each $j=1,...,n$ an open neighborhood $U_j\ni p_j$ and a homeomorphism $\varphi_j:U_j\to V_j \subset\hat M_0$, with $\varphi_j(p_j)=0$ satisfying the following:
\begin{enumerate}[label=(\roman*)]
\item Let $\gamma:(a,b)\to \hat M_0$ be the arc-length parametrization of $\varphi_j(\partial M\cap U_j)$ which is positively oriented and satisfies $\gamma(0)=0$. Then $\pi(\gamma|_{(a,0]})$ and $\pi(\gamma|_{[0,b)})$ are smooth and
$$\alpha_j\pi=\lim_{t\to 0-}\arg\gamma(t)-\lim_{t\to 0+}\arg\gamma(t),$$
so that $V_j$ has an interior angle $\alpha_j\pi$ at $0$ (with respect to the Euclidean metric). 
\item The pull-back $(\varphi_j^{-1})^\ast g$ can be expressed as $e^{2\sigma_j}dz$ where $\sigma_j\in C^\infty(V_j^\circ)$ and all partial derivatives of $\sigma_j$ extend continuously to $\partial V_j$.
\end{enumerate}
\end{definition}
\begin{remark}\label{rmk:differentdefs}
First of all we remark that, if $\alpha_j<2$ for all $j$, then we may replace the usage of $\hat M_0$ by $\C$ (or $\R^2$). In~\cite{NRS}, a curvilinear polygonal domain is defined slightly differently: there, it is a compact subset of a smooth Riemannian surface, which has piecewise smooth boundary where the corner angles lie in the open interval $(0,2\pi)$. A surface which is a curvilinear polygonal domain in the sense of~\cite{NRS} is trivially a curvilinear polygonal domain in the sense of Definition~\ref{def:cpd}. Conversely, a curvilinear polygonal domain in the sense of Definition~\ref{def:cpd}, with $\alpha_j<2$ for all $j$, can be made into a curvilinear polygonal domain in the sense of~\cite{NRS} by constructing a slightly larger Riemannian surface with smooth boundary $\tilde M$: As mentioned above we can replace the usage of $\hat M_0$ with $\C$ when $\alpha_j<2$. The surface $M$ can be extended close to the corners by (abstractly) extending the patches $\varphi_j:U_j\to V_j\subset \C$ to patches $\tilde \varphi_j:\tilde U_j\to \tilde V_j\subset \C$ such that $\tilde V_j$ has smooth boundary. This yields a surface $\tilde M$ with smooth boundary. We may then extend the metric $g$ to $\tilde M\setminus M$ by in each coordinate $\tilde V_j$ extending $\sigma_j$ to $\tilde \sigma_j\in C^\infty(\tilde V_j)$ (such an extension is possible by condition (ii) of Definition \ref{def:cpd} and the regularity of $\partial M$, see, e.g., \cite[Theorem VI.5]{S70}).
\end{remark}
\begin{definition}\label{def:smooth}Let $(M,g,(p_j),(\alpha_j))$ be a curvilinear polygonal domain. We say that $\psi:M\to \R$ is smooth, $\psi\in C^\infty(M,g,(p_j),(\alpha_j))$, if $\psi\in C^\infty(M\setminus\{p_1,...,p_n\})$ and if there is a choice of $(\varphi_j,U_j)$ as in Definition \ref{def:cpd} such that all partial derivatives of $\psi\circ\varphi_j^{-1}$ extend continuously to $\partial V_j$. 
\end{definition}
\begin{remark}
It follows directly from Definitions~\ref{def:cpd} and~\ref{def:smooth} that, if $(M,g,(p_j),(\alpha_j))$ is a curvilinear polygonal domain and $\psi\in C^\infty(M,g,(p_j),(\alpha_j)),$ then $(M,e^{2\psi}g,(p_j),(\alpha_j))$ is also a curvilinear polygonal domain, since, in the notation of Definition~\ref{def:cpd} $$(\varphi^{-1}_j)^\ast e^{2\psi}g = e^{2(\psi(\varphi_j^{-1})+\sigma_j)}dz.$$
\end{remark}
In the proof of Theorem \ref{thm:heattrace} we will use that, for each corner $p_j$ of a curvilinear polygonal domain $(M,g,(p_j),(\alpha_j))$ one can construct a smoothly bounded Jordan domain $V_{j,1}\subset \hat M_0$  with the following properties. Let $\Gamma_1$ denote the positively oriented boundary of $V_{j,1}$. Then we may choose $V_{j,1}$ so that $\Gamma_1$ is a smooth and closed extension of $\gamma_{(a,0]}$, see Figure \ref{fig:Vjk}.
\begin{figure}
\centering
\includegraphics[width=0.35\linewidth]{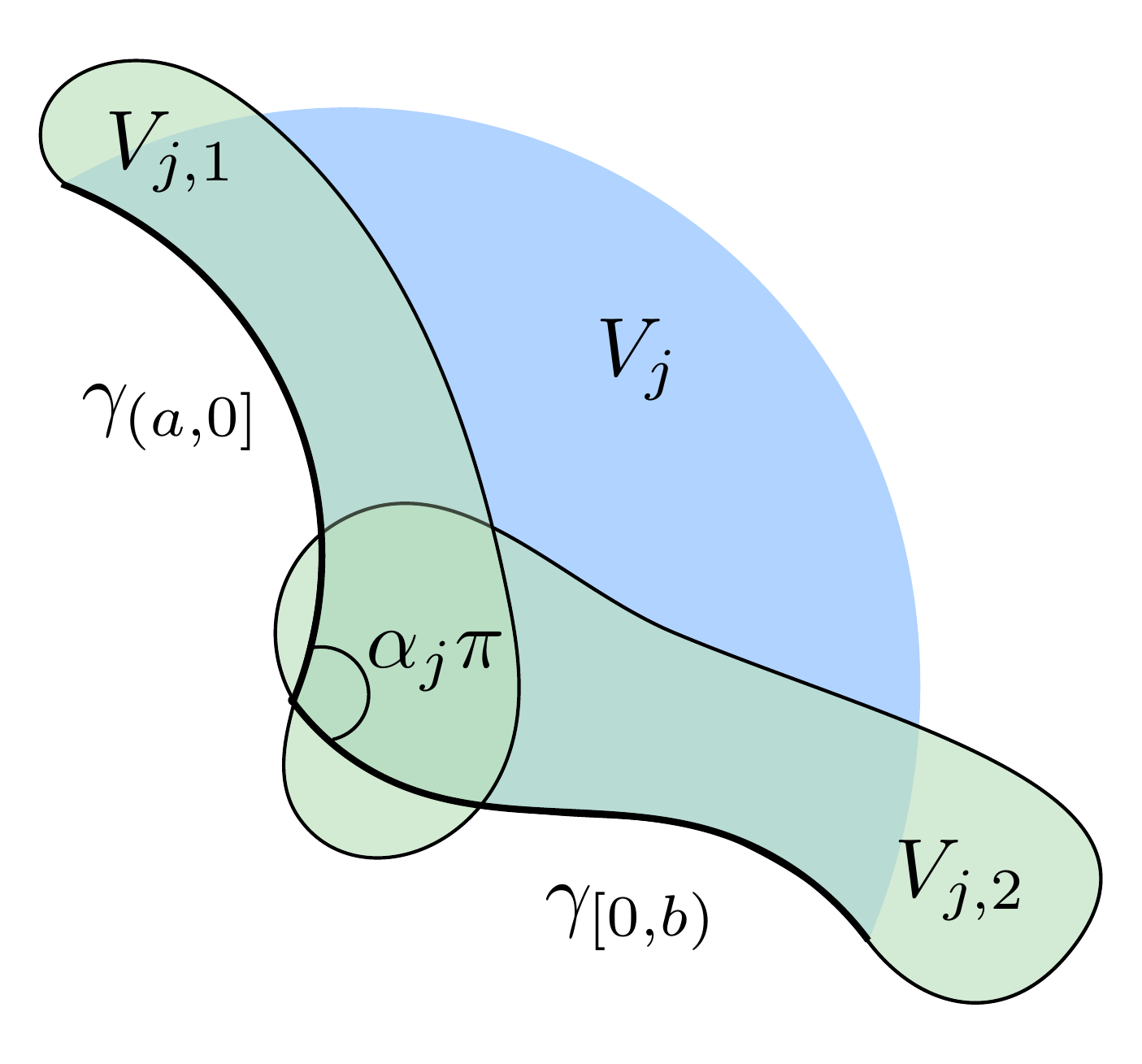}
\caption{Illustration of $V_{j,1},V_{j,2}\subset \hat M_0$.\label{fig:Vjk}}
\end{figure}
 If $\overline V_{j,1}\subset V_j$ then we have directly that $\sigma_j|_{\overline V_{j,1}}\in C^\infty(\overline V_{j,1})$. If not, then $\sigma_j$ can be extended to $\overline V_{j,1}$ so that $\sigma_j|_{\overline V_{j,1}}\in C^\infty(\overline V_{j,1})$. Similarly, if $\psi\in C^\infty(M,g,(p_j),(\alpha_j))$ then $\psi\circ\varphi_j^{-1}$ can be extended to $\overline V_{j,1}$ so that $\psi\circ\varphi_j^{-1}|_{V_{j,1}}\in C^\infty(V_{j,1})$. In a similar manner, one can construct $V_{j,2}$ which has $\gamma_{[0,b)}\subset\partial V_{j,2}$ and both $\sigma_j$ and $\psi\circ\varphi_j^{-1}$ can be extended to be smooth on $\overline V_{j,2}$ (note however, that the extensions may not agree on $\overline V_{j,1}\cap \overline V_{j,2}$). 
\subsection{The Laplace-Beltrami operator and Brownian motion}\label{section:laplacian}
In this section, we recall some basics about the Laplace-Beltrami operator and Brownian motion. The reader is referred to~\cite[Chapter I and VII]{Chavel} and~\cite[Chapter V]{Ikeda} for detailed treatments of the Laplace-Beltrami operator, heat kernel, and Brownian motion on Riemannian manifolds.\par
Let $(M,g,(p_j),(\alpha_j))$ be a curvilinear polygonal domain. Let $C^\infty_c(M^\circ)$ denote the space of smooth functions with compact support on $M^\circ$, and $C^\infty_0(M)$ denote the space of smooth functions $f$ on $M^\circ$ which can be continuously extended to $\partial M$ with $f|_{\partial M}\equiv 0$. We consider $C^\infty_c(M^\circ)$ as a subset of $L^2(M,\Vol_g)$ and denote the $L^2$-inner product by $\langle\cdot,\cdot\rangle$. The Dirichlet Laplace-Beltrami operator is, a priori, defined on $C^\infty_c(M^\circ)$ (in local coordinates) by
$$\Delta_{(M,g)}=-\sum_{i,j=1}^2\frac{1}{\sqrt{\det g}}\partial_i g^{ij}\sqrt{\det g}\partial_i.$$
Here $(g^{ij})$ denotes the inverse of $(g_{ij})$. 
With this sign convention $\Delta_{(M,g)}$ is symmetric and positive definite. Since $C^\infty_c(M^\circ)$ is dense in $L^2(M,\Vol_g)$ this allows us to take the Friedrichs extension of $\Delta_{(M,g)}$, which we for convenience also denote by $\Delta_{(M,g)}$, so that $\Delta_{(M,g)}$ becomes self-adjoint.
\par
Brownian motion on $(M,g)$, stopped upon hitting $\partial M$ is the Markov process with $-\Delta_{(M,g)}$ as its infinitesimal generator. The transition density function of Brownian motion is the heat kernel, $H_{(M,g)}(t;p,q)\in C^\infty((0,\infty)\times M^\circ\times M^\circ)$.
\begin{remark}Typically, Brownian motion is defined to be the Markov process with $-\tfrac{1}{2}\Delta_{(M,g)}$ as its infinitesimal generator, so that the transition density function is 
$$P(t;p,q)=H_{(M,g)}(t/2;p,q)$$ with $H_{(M,g)}$ as above. We use the convention above for convenience.  
\end{remark}
The heat kernel is the fundamental solution to the Dirichlet heat equation, that is, for every $f\in C(M)$ we have that 
$$u(t,p)=\int_{M}H_{(M,g)}(t;p,q)f(q)d\Vol_g(q)\in C^\infty( (0,\infty)\times M^\circ\times M^\circ)$$
solves
\begin{equation}\label{eq:intitialdirichletheat}\begin{cases}
\partial_t u(t,p) = -\Delta_{(M,g)} u(t,p),&\ t>0,\ p\in M^\circ,\\
u(t,p) = 0,&\ t>0,\ p\in\partial M,\\
\lim_{t\to 0+}u(0,p) = f(p),&\ p\in M^\circ.
\end{cases}\end{equation}
If we define $P_t:C(M)\to C^\infty_0(M)$ by $P_t f:=u(t,\cdot)$, then the family $(P_t)_{t}$ forms a semi-group with $P_{t_1}P_{t_2}=P_{t_1+t_2}$. From this, we deduce that
$$\langle f,P_t f\rangle = \langle P_{t/2}f,P_{t/2} f\rangle \geq 0.$$
Thus, $P_t$ is positive semi-definite. Since $\int_{M}H_{(M,g)}(t;p,p)d\Vol_g<\infty$, the operator $P_t$ can be extended to a continuous, positive definite, compact, and self-adjoint operator on $L^2(M,\Vol_g)$. From the spectral theorem and the semigroup property it then follows that there is an orthonormal set of eigenfunctions $(\phi_n)_{n\geq 1}\in C^\infty_0(M)$ and corresponding eigenvalues $(e^{-t\lambda_n})_{n\geq 1}$ where $$0<\lambda_1\leq \lambda_2\leq ...,\ \lim_{n\to\infty}{\lambda_n}=\infty.$$
The heat kernel therefore has the representation
$$H_{(M,g)}(t;p,q)=\sum_{n\geq 1}e^{-t\lambda_{n}}\phi_n(p)\phi_n(q).$$
Furthermore, it can be deduced from~\eqref{eq:intitialdirichletheat} that $\Delta_{(M,g)}\phi_n=\lambda_n\phi_n$. We can therefore identify $P_t = e^{-t\Delta_{(M,g)}}$.
The trace of the heat kernel, the heat trace, is 
$$\int_M H_{(M,g)}(t;p,p)d\Vol_g = \Tr(e^{-t\Delta_{(M,g)}})=\sum_{n\geq 1}e^{-t\lambda_n}.$$
Finally, we recall Weyl's law (\cite{Weyl},~\cite[Section VII.3]{Chavel})
\begin{equation}\lambda_n\sim \frac{4\pi n}{\Vol_g(M)},\quad n\to\infty.\label{eq:weyl}\end{equation}
\subsection{Basic properties of the heat kernel}
In this section, we recall a few basic properties of the heat kernel. Here $c_1,c_2,...$ will denote positive constants.\par
\paragraph{Domain monotnicity.}Suppose $(M,g)$ is a Riemannian surface and $U\subset M$ open. Then 
\begin{equation}H_{(M,g)}(t;p,q)\geq H_{(U,g)}(t;p,q),\ x,y\in U,\label{eq:monotonicity}\end{equation}
since the density at $q$ of a Brownian motion on $(M,g)$ started at $p$ and run for time $t$ will become smaller if it is stopped upon exiting $U$. See, e.g.,~\cite[Chapter VIII]{Chavel}. Using the convention $H_{(U,g)}(t;p,q)=0$ if $(p,q)\notin U\times U$,~\eqref{eq:monotonicity} holds for all $(p,q)\in M\times M$.
\paragraph{Local Gaussian bounds.}If $(M,g)$ is a smooth Riemannian surface with boundary, $(U,\varphi)$ is a smooth coordinate, and $K\subset U$ is compactly contained in $U$, then there exists constants $c_1$ and $c_2$ such that
$$H_{(M,g)}(t;\varphi^{-1}(x),\varphi^{-1}(y))\leq \frac{c_1}{t}e^{-c_2|x-y|^2/t}, $$
for sufficiently small $t$ and $x,y\in K$ (see~\cite[Equation (3.6)]{McKeanSinger}). Similarly, consider the heat kernel on $(M_0,dz)$. Since Brownian motion on $M_0$ can be constructed by simply lifting Brownian motion on $\C\setminus\{0\}$ by $\pi:M_0\to\C$, we have
$$H_{(M_0,dz)}(t;z,w)\leq H_{(\C,dz)}(t;\pi(z),\pi(w))=\frac{1}{4\pi t}e^{-|\pi(z)-\pi(w)|^2/4t}.$$
If $z,w\in M_0$ are such that $\dist_{M_0}(z,w)>|\pi(z)-\pi(w)|$ then the shortest path on $\hat M_0$ from $z$ to $w$ is the broken line-segment through $0$ (so that $\dist_{M_0}(z,w)=|z|+|w|$). In this case, we obtain
$$H_{(M_0,dz)}(t;z,w)\leq \min\Big(H_{\C,dz}(t;\pi(z),0),H_{\C,dz}(t;\pi(z),0)\Big)\leq \frac{1}{4\pi t}e^{-\dist_{dz}(z,w)^2/t}.$$
As a conformal scaling $e^{2\sigma}dz$, for $\sigma$ smooth and bounded, simply alters the time-parametrization of the Brownian motion we have
$$H_{(M_0,e^{2\sigma}dz)}(t;z,w)\leq  \frac{c_3}{t}e^{-c_4\dist_{e^{2\sigma}dz}(z,w)^2/t},$$
in that case as well. 
\paragraph{Kac's locality principle.} The short time behavior of the heat kernel, $H_{(M,g)}(t;p,q)$, is governed, if $p$ and $q$ are close, by the geometry of $M$ close to $p$. Heuristically, one can argue that a Brownian motion started at $p$ (or more precisely a Brownian bridge from $p$ to $q$) is unlikely to exit a fixed neighborhood $U$ of $p$ and $q$ within a small time $t$. Therefore, the Brownian motion does not ``feel'' the geometry outside $U$ and hence 
$$H_{(M,g)}(t;p,q)\sim H_{(U,g)}(t;p,q)\quad \text{as }\quad t\to 0+.$$
We now make this precise. Let $(M,g,(p_j),(\alpha_j))$ be a curvilinear polygonal domain and fix an open subset $U\subset M$. For $p,q\in U$ 
$$H_{(M,g)}(t;p,q)-H_{(U,g)}(t;p,q)=\lim_{\delta\to 0+}\frac{\P^{p}_{(M,g)}[\tau_{U}<t<\tau_{M},\ B_t\in B(q,\delta)]}{\Vol_{g}(B(q,\delta))},$$
where $\P^p_{(M,g)}$ is the law of a Brownian motion on $(M,g)$ started at $p$ and stopped at $$\tau_{M}= \inf\{t:B^p_t\in\partial M\}.$$ Let $\tilde U$ be open and compactly contained in $U$. Then, $\dist_g(\tilde U,M\setminus U)=\vare>0$.
Define, $\tau_1=\tau_{U}$, and for $n\geq 1$
$$\sigma_{n}=\inf\{t>\tau_{n}:\dist_g(B_t,\tilde U)=\vare/2\}\quad\text{ and }\quad\tau_{n+1}=\inf\{t>\sigma_{n}:B_t\in\partial U\}.$$
By the strong Markov property and the local Gaussian bounds
\begin{align*}
&\lim_{\delta\to 0+}\frac{\P_{(M,g)}^p[\tau_{U}<t<\tau_M:B_t\in B(q,\delta)]}{\Vol_g(B(q,\delta))}\\
=&\lim_{\delta\to 0+}\frac{\sum_{n\geq 1}\P_{(M,g)}^p[\sigma_n<t<\tau_{n+1}:B_t\in B(q,\delta)]}{\Vol_g(B(q,\delta))}
\\
=&\sum_{n\geq 1}\int_{\{\sigma_n<t \}}H_{(U,g)}(s;r,q)d\P_{(M,g)}^p[B_{\sigma_{n}}=r,\ \sigma_n=t-s]\\
\leq & \frac{c_{5}}{t}e^{-c_{6}\vare^2/4t}\sum_{n\geq 1}\P^x_{(M,g)}[\sigma_n<t].
\end{align*}
On the event $\sigma_n<t$ the Brownian motion has travelled, back from $\partial U$ to $\{r:\dist_g(r,\tilde U)=\vare/2\}$, $n$ times. Since the probability that a Brownian motion exits a ball of radius $\vare/2$ within time $t$ can be bounded above by $c_{7}t/\vare^2$, we find that, for $p,q\in\tilde U$
$$H_{(M,g)}(t;p,q)-H_{(U;g)}(t;p,q)\leq c_8t^{-1}e^{-c_{9}\vare^2/t}$$
for sufficiently small $t$ and $\vare^2>c_10t$. By a similar argument, one finds 
$$H_{(M,g)}(t;p,q)\leq c_8t^{-1}e^{-c_{9}\vare^2/t}$$
for sufficiently small $t$ and $\vare^2>c_10t$, when $p\in \tilde U$ and $q\notin U.$
\paragraph{Global Gaussian bounds.}
By combining the local Gaussian bounds and the locality principle one can obtain global Gaussian bounds in a curvilinear polygonal domain $(M,g,(p_j),(\alpha_j))$: for sufficiently small $t$ and all $p,q\in M$
\begin{equation}H_{(M,g)}(t;p,q)\leq\frac{c_{10}}{t}e^{-c_{11} \dist_g(p,q)^2/t}.\label{eq:gaussian}\end{equation}
In~\cite{Davies}, the author provides bounds on the time derivatives of the heat kernel, given (local) bounds on the heat kernel itself. By a direct application of~\cite[Corollary 5]{Davies}, we obtain
\begin{equation}\partial_t H_{(M,g)}(t;p,q)=\frac{c_{12}}{t^2}e^{-c_{13} \dist_g(p,q)^2/t},\label{eq:Davies}\end{equation}
for all $p,q\in M^\circ$ and sufficiently small $t$.
\subsection{The McKean-Singer construction}\label{section:mckeansinger}
In~\cite{McKeanSinger}, the heat kernel of a smooth Riemannian manifold with boundary represented by a series. As the proof of Theorem~\ref{thm:heattrace} relies heavily on estimates from \cite{McKeanSinger} of the heat kernel, obtained via the series representation, we briefly summarize the set-up here. Consider $\R^2$ endowed with a smooth metric $g_{ij}$ such that $g_{ij}(x)=\delta_{ij}$ for $|x|$ large. Write the Laplace-Beltrami operator as $\Delta_{(\R^2,g)}=\sum_{i,j=1}^2a_{ij}\partial_i\partial_j+\sum_{i=1}^n b_i\partial_i$, that is
$$a_{ij}=-g^{ij},\quad b_i=-\frac{1}{\sqrt{\det g}}\sum_{j=1}^2 \partial_j(g^{ij}\sqrt{\det g}).$$
Further, denote by $Q_{x_0}=\sum_{i,j=1}^2a_{ij}(x_0)\partial_i\partial_j+\sum_{i=1}^n b_i(x_0)\partial_i$, that is, the differential operator obtained by fixing the coefficients of $\Delta_{(\R^2,g)}$ at $x_0$. The (minimal) fundamental solution to $\partial_t u = -Q_{x_0} u$ is
$$H_{x_0}(t;x,y)=\frac{1}{4\pi t}\exp\bigg(-\sum_{i,j=1}^2 g_{ij}(x_0)(y_i-x_i+b_i(x_0)t)(y_j-x_j+b_j(x_0)t)/4t\bigg).$$ 
Denote by $H^0(t;x,y)=H_{y}(t;x,y)$. Then Duhamel's principle gives 
\begin{align*}H_{(\R^2,g)}(t;x,y)-&H^0(t;x,y)\\ = \int_0^t & \partial_s \int_{M}H_{(\R^2,g)}(s;x,z) H^0(t-s;z,y)\sqrt{\det g}dz ds\\
=\int_0^t &\int_{M} H_{(M,g)}(s;x,z)(Q_{y,z}-\Delta_{(\R^2,g),z})H^0(t-s;z,y)\sqrt{\det g}dzds,
\end{align*}
where the subscript $z$ on the final line indicates differentiation with respect to $z$. 
Using the notation
\begin{align*}&f\sharp g (t;x,y):=\int_{0}^t\int_M f(s;x,z)g(t-s;z,y)d\Vol_{g}(z)ds\\
& f\sharp_0 g:= f,\qquad f \sharp_n g := (f\sharp_{n-1} g)\sharp g,\ n=1,2,3,..., 
\end{align*}
and defining
$$G(t;x,y):=(Q_{y,x}-\Delta_{(\R^2,g),x})H^0(t;x,y)$$
the above can be expressed as
\begin{align}H_{(\R^2,g)} = H^0 + H_{(\R^2,g)}\sharp G.\label{eq:duhamel}\end{align}
By iterating, we find the formal representation 
\begin{equation}H_{(\R^2,g)} = \sum_{n=0}^\infty H^0\sharp_n G.\label{eq:levisum}\end{equation}
Using that 
$$H^0(t;x,y)\leq \frac{c_1}{t}e^{-c_2|x-y|^2/t},\qquad |G(t;x,y)|\leq \frac{c_3}{t^{3/2}}e^{-c_4|x-y|^2/t}$$ an explicit computation shows
$$|H^0\sharp_n G (t;x,y)|\leq c_5^n[(n/2)!]^{-1}t^{n/2-1}e^{-c_6|x-y|^2/t}.$$
Hence, the right-hand side of~\eqref{eq:levisum} converges and as a result~\eqref{eq:levisum} holds. 
\subsection{The \texorpdfstring{$\zeta$}{}-regularized determinant of the Laplacian}\label{section:zdet}
Let $(M,g,(p_i),(\alpha_i))$ be a curvilinear polygonal domain and consider the Friedrichs extension of the Dirichlet Laplace-Beltrami operator on $M$, $\Delta_{(M,g)}$. Then, as we saw in Section~\ref{section:laplacian}, there is an orthonormal basis of $\dom(\Delta_{(M,g)})$ of eigenfunctions of $\Delta_{(M,g)}$ with corresponding eigenvalues 
$$0<\lambda_1\leq \lambda_2\leq ...$$
satisfying Weyl's law~\eqref{eq:weyl}. We define,
$$\zeta_{(M,g)}(s)=\sum_{n\geq 1}\lambda_n^{-s},\quad \Re s>1,$$
where the right-hand side converges by Weyl's law. Following Ray and Singer \cite{RaySinger}, we express the spectral $\zeta$-function using the heat trace
$$\zeta_{(M,g)}(s)=\frac{1}{\Gamma(s)}\int_0^\infty t^{s-1}\Tr(e^{-t\Delta_{(M,g)}})dt.$$
The asymptotic expansion of the heat trace
\begin{equation}\Tr(e^{-t\Delta_{M,g}}) = a_0 t^{-1}+a_{-1/2}t^{-1/2}+a_0+ O(t^{q}),\quad \text{as }t\to 0+,\label{eq:heattraceanonym}\end{equation}
from Theorem~\ref{thm:heattrace}, can then be used to analytically continue $\zeta_{(M,g)}$. For $\Re s>1$, we have
\begin{align*}
\zeta_{(M,g)}(s)=&\frac{1}{\Gamma(s)}\int_0^\infty t^{s-1}\Tr(e^{-t\Delta_{(M,g)}})dt\\
=&\frac{1}{\Gamma(s)}\bigg(\frac{a_{-1}}{s-1}+\frac{a_{-1/2}}{s-1/2}+\frac{a_0}{s}\bigg)\\
&+\frac{1}{\Gamma(s)}\int_0^1 t^{s-1}(\Tr(e^{-t\Delta_{(M,g)}})-a_0 t^{-1}-a_{-1/2}t^{-1/2}-a_0)dt\\& + \frac{1}{\Gamma(s)}\int_1^\infty t^{s-1}\Tr(e^{-t\Delta_{(M,g)}})dt.\end{align*}
Since $1/\Gamma(s)=s+O(s^2)$ is entire,~\eqref{eq:heattraceanonym} shows that the right-hand side is analytic on $$\{s\in\C:\Re s>-1/2,\ s\neq 1,1/2\}.$$ Thus, the right-hand side above provides an analytic extension of $\zeta_{(M,g)}$ to the twice punctured half-plane. In particular, the $\zeta$-regualrized determinant of $\Delta_{(M,g)}$, $\zdet\Delta_{(M,g)}:=e^{-\zeta_{(M,g)}'(0)}$, is well-defined.
\section{Short time asymptotic expansion of the heat trace}\label{section:heattrace}
In this section, we prove Theorem~\ref{thm:heattrace} using a patchwork technique. This is a classical method for approximating the heat trace, see, e.g.,~\cite{Kac,McKeanSinger,VDBS,LR}. The estimates of the heat trace in a smooth Riemannian surface with boundary from~\cite{McKeanSinger}, and the estimates of the heat trace in a flat and straight wedge from~\cite{VDBS} are at the foundation of the proof. Our main task is to combine the two and to handle the non-flat and non-straight behavior locally at the corners.\par 
We first provide three lemmas which give bounds for the diagonal of the heat kernel (or heat trace) at points far from the boundary (Lemma~\ref{lemma:interior}), points close to the boundary but far from a corner (Lemma~\ref{lemma:boundary}), and points close to a corner (Lemma~\ref{lemma:corner}).\par 
Throughout this section $c_1,c_2,c_3,...$, $t_1,t_2,t_3,...$ denotes positive constants which are named consistently within each statement and proof, but not consistent between different statements and proofs.
To simplify notation we write, with $g_u$ as in Theorem \ref{thm:heattrace},
$$H_u=H_{(M,g_u)},\quad d\Vol_{u}=d\Vol_{g_u},\quad d\ell_{u}=d\ell_{g_u},\quad K_u=K_{g_u},\quad k_{u}=k_{g_u},\quad \partial_{n_u}=\partial_{n_{g_u}}.$$
\begin{lemma}\label{lemma:interior}Let $(M,g_0,(p_j),(\alpha_j))$, $\sigma,$ and $g_u$ be as in Theorem~\ref{thm:heattrace}. For every compact interval $I$ there exists positive constants $c_1,\ c_2,\ c_3$, and $t_1$ such that 
\begin{equation}
\bigg|H_u(t;p,p)-\frac{1}{4\pi t} -\frac{1}{12\pi}K_{u}(p)\bigg|\leq c_1 t + \frac{c_2}{t}e^{-c_3d^2_p/t},\label{eq:interiorbound}
\end{equation}
for all $t\in(0,t_1)$, $u\in I$, and $p$ such that $\dist_{g_u}(p,\partial M) \geq \sqrt{t}$. Here $d_p=\dist_{g_u}(p,\partial M)$.
\end{lemma}
\begin{remark}It follows immediately from~\cite[Section 4]{McKeanSinger} that~\eqref{eq:interiorbound} holds point-wise on $M$ for each fixed $u$.
\end{remark}
\begin{proof}
In~\cite[Section 4]{McKeanSinger} the diagonal of the heat kernel for a smooth metric (which coincides with the Euclidean metric outside some compact set) on $\R^2$ (and in general $\R^n$) is estimated. For a fixed point $p\in \R^2,$ this is done by changing to geodesic normal coordinates with respect to $p\mapsto 0$, which has the effect that $$g_{ij}(x)=\delta_{ij}+\sum_{k,l=1}^2\frac{1}{3}R_{ikjl}(0)x_kx_l+E_{ij}(x)$$
where $|E_{ij}(x)|\leq c_4|x|^3$ and $|\partial_k E_{ij}(x)|\leq c_5|x|^2$ for $|x|\leq c_6$ and some $c_6>0$. Fix a smooth function $\eta:[0,\infty)\to [0,1]$ such that
$$
\eta(r)=1,\quad \forall r\in [0,c_6/2]\quad\text{and}\quad \eta(r)=0,\quad \forall r\in [c_6,\infty),
$$
and consider 
$$\tilde g_{ij}(x)=\delta_{ij}+\eta(|x|)\bigg(\frac{1}{3}\sum_{k,l=1}^2 R_{ikjl}(0)x_lx_m + E_{ij}(x)\bigg).$$
Then the computation of~\cite[Section 5]{McKeanSinger} shows that the corresponding heat kernel $\tilde H$ satisfies
\begin{equation}
\Big|\tilde H(t;0,0)-\frac{1}{4\pi t}+\frac{1}{12\pi}R_{1212}(0)\Big|\leq c_7t,\quad \forall t\in(0,t_2]\label{eq:McKeanSigerInterior}
\end{equation}
where $c_7$ and $t_2$ depend only on $c_4,\ c_5,\ \eta$, and $K_{\tilde g}(0)=R_{1212}(0)$ (the latter is not explicitly stated but is seen upon examining the proof of~\cite[Equation (4.2)]{McKeanSinger}).
\par 
Let $g_{0,ij}^p$ denote the metric $g_0$ in normal coordinates with respect to $p\in M^\circ$. Observe that, by regularity of the metric on the smooth part of the boundary and at the corner, the metric can be extended smoothly across the boundary. At corners, this is slightly subtle. If $\alpha_j<2$ then the metric can be extended across the corner as explained at the end of Remark \ref{rmk:differentdefs}. If $\alpha_j\geq 2$, we instead employ several (but finitely many) different extensions across the corner. This can be done by choosing finitely many half-planes $$\H_\alpha=\{(r,\theta)\in M_0:\theta\in(\alpha\pi,(\alpha+1)\pi)\}$$ which cover $V_j$, and then extending $\sigma_j|_{\H_\alpha}$ across the half-plane (strictly speaking we first project onto $\C$ and then extend across the projected half-plane). Let $\hat g_{0,ij}^p$ denote the normal coordinate centered at $p$ in (one of) the extension(s) of the metric $g$. By smoothness and compactness, there is a uniform lower bound on the injectivity radius, say $2c_6$. That is, we assume that $\hat g_{0,ij}^p$ is defined on $B(0,2c_6)$ for all $p\in M$. Furthermore, since the coefficients in the expansion of $\hat g_{0,ij}^p$ depend smoothly on $p$ (for each of the finitely many extensions) there exist constants $c_8$ and $c_9$ such that
$$\hat g_{0,ij}(x)=\delta_{ij}+\frac{1}{3}\sum_{k,l=1}^2 R_{0,ikjl}^p(0)x_lx_m + E^p_{0,ij}(x),$$
where $|E^p_{0,ij}(x)|\leq c_{8}|x|^3$, $|\partial_k E^p_{0,ij}(x)|\leq c_{9}|x|^2$ for all $p\in M$. Moreover, the same procedure can be carried out for $g_u$ for each $u\in\R$ and since $\hat g_u$ can be made to depend smoothly on $u$ the constants $c_6$, $c_{8}$, and $c_{9}$, can be set so that we have $$|E^p_{u,ij}(x)|\leq c_{8}|x|^3,\quad |\partial_k E^p_{u,ij}(x)|\leq c_{9}|x|^2$$ for all $|x|\leq c_6$, and $u\in I$. Hence,~\eqref{eq:McKeanSigerInterior} and the locality principle implies that
$$\Big|H_u(t;p,p)-\frac{1}{4\pi t}-\frac{1}{12\pi}K_{u}(p)\Big|\leq c_{1}t+\frac{c_{2}}{t}e^{-c_{3}d^2_p/t},\quad \forall t\in(0,t_1],$$
since the metric $\tilde g_{u,ij}^p(x)$ agrees with $\hat g_{u,ij}^p(x)$ within $B(0,c_6/2)$ (independent of $p$), and $\hat g^p_{u,ij}$ agrees with $g_{u,ij}^p$ within $B(0,c_{10}d_p)$ for some $c_{10}>0$.
\end{proof}
\begin{lemma}\label{lemma:boundary}Let $(M,g_0,(p_j),(\alpha_j))$, $\sigma,$ $\psi$, and $g_u$ be as in Theorem~\ref{thm:heattrace}. Consider a rectangular boundary patch, that is, $U\subset M\setminus\{p_1,...,p_n\}$ open and a smooth homeomorphism $\varphi:U\to [0,L)\times(a,b)$, satisfying $g_{12}(0,x_2)=0$ for all $x_2\in (a,b)$. Fix an interval $J$ compactly contained in $(a,b)$, and a second compact interval $I$. Then, there exists constants such $c_1,...,c_7$ so that, for all rectangles $R=[0,\vare)\times J$, with $\vare<L$, and $u\in I$
\begin{align}\begin{split}
&\bigg|\int_{\varphi^{-1}(R)}\psi H_u(t;p,p)d\Vol_{u} -\bigg( \frac{1}{4\pi t}\int_{\varphi^{-1}(R)}\psi\Big(1+\frac{t}{3}K_{u}\Big) d\Vol_{u}\\&-\frac{1}{8\sqrt{\pi t}}\int_{\varphi^{-1}(B)}\psi\Big(1-\frac{2}{3}\sqrt{\frac{t}{\pi}}k_{u}\Big) d\ell_{u}+\frac{1}{8\pi}\int_{\varphi^{-1}(B)}\partial_{n_{u}}\psi\sigma d\ell_{u}\bigg)\bigg|\\ &\leq c_1t^{1/2}+c_2\vare+c_3\vare^2t^{-1/2}+\frac{c_4}{t}e^{-c_5d^2_\vare/t}+\frac{c_6}{t}e^{-c_7\vare^2/t},\ \end{split}\label{eq:boundarybound}
\end{align}
for sufficiently small $t$, where $d_\vare=\min(L-\vare,\dist(J,\partial(a,b))$ and $B=\{0\}\times J$.
\end{lemma}
\begin{remark}It follows directly from~\cite[Section 5]{McKeanSinger}~\eqref{eq:boundarybound} holds for $\psi\equiv 1$ and a fixed $u$. Similar to Lemma~\ref{lemma:interior}, it follows by a careful read of~\cite[Section 5]{McKeanSinger} that~\eqref{eq:boundarybound} holds $\psi\equiv 1$ for $u$ in a bounded interval since the metrics $g_u$ depend smoothly on $u$. This type of expansion is considered well known, also for $\psi\not\equiv 1$, (see, e.g.,~\cite{AKR}), but we were not able to find a statement of this type when the integral is restricted to a rectangle (rather than the integral being over the entire space) and therefore we provide a proof here.
\end{remark}
\begin{proof}We follow~\cite[Section 5]{McKeanSinger}. Let $g_{u,ij}$ be the metric $g_u$ expressed in the $\varphi$ coordinate and extend $g_{u,ij}$ smoothly to $\R^+\times\R$ so that $g_{u,ij}(x)=\delta_{ij}(x)$ for $|x|$ large, and then to $\R\times\R$ by
$$g_{u,11}(x)=g_{u,11}(x^\ast),\ g_{u,12}(x)=-g_{u,12}(x^\ast),\ g_{u,22}(x)=g_{u,22}(x^\ast),$$
where $(x_1,x_2)^\ast=(-x_1,x_2)$. Let $\hat H_u$ be the (minimal) heat kernel corresponding to this extension, and note that 
$$\tilde H_u(t;x,y)=\hat H_u(t;x,y)-\hat H_u(t;x,y^\ast)$$
is the Dirichlet heat kernel with respect to $g_{u,ij}$ on $\R^+\times\R$. By the locality principle 
$$|\tilde H_u(t;x,x)-H_u(t;\varphi^{-1}(x),\varphi^{-1}(x))|\leq \frac{c_8}{t}e^{-c_9d^2_\vare/t},$$
for all $x\in R$, and therefore
\begin{equation}\bigg|\int_R\hat\psi\Big(\tilde H_u(t;x,x)-H_u(t;\varphi^{-1}(x),\varphi^{-1}(x))\Big)\sqrt{\det{g_u}}dx\bigg| \leq \frac{c_{10}\vare}{t}e^{-c_{11}d^2/t},\label{eq:levisumlocality}\end{equation}
Recall from Section \ref{section:mckeansinger}, that $\hat H_u(t;x,y)$ has the representation
$$\hat H_u(t;x,y)=\sum_{n= 0}^\infty H^0_u\sharp_n G_u(t;x,y)$$
where 
\begin{equation}|H^0_u\sharp_n G_u(t;x,y)|\leq \frac{c_{12}^n}{(n/2)!}t^{n/2-1}e^{-c_{13}|x-y|^2/t}.\label{eq:levisumbound}\end{equation}
This bound can be made uniform in $u\in I$ since $g_{u,ij}(x)=e^{2u\sigma(\varphi^{-1})}g_{0,ij}(x)$ for $x\in [0,L)\times (a,b)$. Similarly, all of the bounds that we state below can be made locally uniform in $u$ for the same reason. In the ordo notation below we always consider $\vare\to 0+$ and $t\to 0+$.
The bound~\eqref{eq:levisumbound} gives
$$\sum_{n=2}^\infty H^0_u\sharp_n G_u(t;x,y)=O(1)$$
and hence
\begin{equation}\bigg|\int_R \hat\psi\bigg(\sum_{n=2}^\infty H^0_u\sharp_n G_u(t;x,x)-\sum_{n=2}^\infty H^0_u\sharp_n G_u(t;x,x^\ast)\bigg)\bigg|\sqrt{\det g_u}dx\leq O(\vare),\label{eq:ngeqtwoterms}\end{equation}
where $\hat\psi=\psi\circ\varphi^{-1}$.
It remains to estimate integrals involving $H^0_u(t;x,x)$, $H^0_u(t;x,x^\ast)$, $H^0_u\sharp G_u(t;x,x)$, and $H^0_u\sharp G_u(t;x,x^\ast)$.
First of all, $H^0_u(t;x,x)=\frac{1}{4\pi t}+O(1)$ and hence 
\begin{align}\int_R  \hat\psi H^0_u(t;x,x)\sqrt{\det g_u}dx= \frac{1}{4\pi t}\int_{\varphi^{-1}(R)} \psi d\Vol_{g_u} +O(\vare).\label{eq:H}\end{align}
We approximate
\begin{equation*}\hat\psi(x_1,x_2)=\hat\psi(0,x_2)-\partial_{x_1}\hat\psi(0,x_2)x_1+O(x_1^2).\end{equation*} 
Below, a superscript $0$ will denote setting the $x_1$-argument to $0$ (this is in accordance with the notation of~\cite{McKeanSinger}). In the proof of~\cite[Equation (5.5a)]{McKeanSinger} it is shown that
\begin{align*}\int_0^\vare H^0_u(t;x,x^\ast)\sqrt{\det g_u}dx_1 =& \frac{1}{4\pi t}\int_{0}^\infty e^{-g_{u,11}^0 x_1^2/t}\Big(1+\frac{\partial_{x_1}g_{u,11}}{g_{u,11}^0}x_1-\partial_{x_1}g_{u,11}\frac{x_1^3}{t}\Big)\sqrt{\det g^0_u}dx_1\\& +O(\vare)+O(t^{1/2})+O(t^{-1}e^{-c_{13}\vare^2/t})\\
=&\frac{1}{8\sqrt{\pi t}}\frac{\sqrt{\det g^0_u}}{\sqrt{g_{u,11}^0}}+O(\vare)+O(t^{1/2})+O(t^{-1}e^{-c_{14}\vare^2/t}).
\end{align*}
Using the same arguments one finds
\begin{align*}\int_0^\vare x_1 H^0_u(t;x,x^\ast)\sqrt{\det g_u}dx_1 =& \frac{1}{4\pi t}\int_{0}^\infty e^{-g_{u,11}^0 x_1^2/t}x_1\sqrt{\det g^0_u}dx_1 +O(\vare)+O(t^{1/2})\\&+O(t^{-1}e^{-c_{14}\vare^2/t})\\
=&\frac{1}{8\pi}\frac{1}{\sqrt{g_{u,11}^0}}\frac{\sqrt{\det g_{u,11}^0}}{\sqrt{g_{u,11}^0}}+O(\vare)+O(t^{1/2})+O(t^{-1}e^{-c_{15}\vare^2/t})\end{align*}
and
$$\int_0^\vare x^2_1H^0_u(t;x,x^\ast)\sqrt{\det g_u}dx_1 = O(\vare)+O(t^{1/2})+O(t^{-1}e^{-c_{16}\vare^2/t}).$$
Thus,
\begin{align}\begin{split}\int_R& \hat\psi H^0_u(t;x,x^\ast)\sqrt{\det g_u}dx \\=& \frac{1}{8\sqrt{\pi t}}\int_J \hat\psi(0,x_2)\frac{\sqrt{\det g^0_u}}{\sqrt{g_{u,11}^0}}dx_2+\frac{1}{8\pi}\int_J\frac{\partial_{x_1}\hat\psi(0,x_2)}{\sqrt{g_{u,11}^0}}\frac{\sqrt{\det g_{u,11}^0}}{\sqrt{g_{u,11}^0}}dx_2\\&+O(\vare)+O(t^{1/2})+O(t^{-1}e^{-c_{16}\vare^2/t})\\
=&\frac{1}{8\sqrt{\pi t}}\int_{\varphi^{-1}(B)} \psi d\ell_{u}-\frac{1}{8\pi}\int_{\varphi^{-1}(B)}\partial_{n_{u}}\psi d\ell_{u}\\ &+O(\vare)+O(t^{1/2})+O(t^{-1}e^{-c_{17}\vare^2/t}).
\end{split}\label{eq:Hast}\end{align}
In a similar spirit, we move on to $H^0_u\sharp G_u(t;x,x)$. Studying the proof of~\cite[Equation (5.5b)]{McKeanSinger} we see that they find 
\begin{align*}\int_{0}^\vare H^0_u\sharp G_u(t;x,x)\sqrt{\det g_u}dx_1 =& -\frac{1}{24\pi}\sqrt{\det g^0_u}\frac{\partial_{x_1}(g_u^{11}\det g_u)}{\det g^0_u}\\ &+O(\vare)+O(t^{1/2})+O(t^{-1/2}e^{-c_{18}\vare^2/t}).\end{align*}
Since, $|H^0_u\sharp G_u(t;x,x)| \leq O(t^{-1/2})$ we obtain
$$ \int_0^\vare x_1|H^0_u\sharp G_u(t;x,x)|\sqrt{\det g_u}dx_1 =O(\vare^2 t^{-1/2}).$$ 
Hence,
\begin{align}\begin{split}\int_R& \hat\psi H^0_u\sharp G_u(t;x,x)\sqrt{\det g_u}dx\\
=&-\frac{1}{24\pi}\int_J \hat\psi(0,x_2)\sqrt{g_{u,11}^0}\frac{\partial_{x_1}(g_u^{11}\det g_u)}{\det g^0_u}\frac{\sqrt{\det g^0_u}}{\sqrt{g_{u,11}^0}}dx_2\\ &+O(\vare)+O(t^{1/2}) +O(t^{-1/2}e^{-c_{17}\vare^2/t})+O(\vare^2 t^{-1/2})\\=&-\frac{1}{12\pi}\int_{\varphi^{-1}(B)} \psi k_{u}d\ell_{u}+O(\vare)+O(t^{1/2})+O(t^{-1/2}e^{-c_{18}\vare^2/t})+O(\vare^2 t^{-1/2}).\end{split}\label{eq:Hsharp}\end{align}
Similarly,~\cite[Equation (5.c)]{McKeanSinger} shows that
$$\int_J\int_0^\vare H^0_u\sharp G_u(t;x,x^\ast)\sqrt{\det g_u} dx_1 dx_2 = O(\vare)+O(e^{-c_{19}\vare^2/t}),$$
and since $|H^0\sharp f(t;x,x^\ast)| =O(t^{-1/2})$ we have
\begin{equation}\int_J\int_0^\vare \hat\psi H^0_u\sharp G_u(t;x,x^\ast)\sqrt{\det g_u} dx_1 dx_2 = O(\vare)+O(e^{-c_{19}\vare^2/t})+O(\vare^2 t^{-1/2}).\label{eq:Hsharpast}\end{equation}
Combining~(\ref{eq:ngeqtwoterms}-\ref{eq:Hsharpast}) we find
\begin{align*}\int_{R} &\hat\psi \hat H_u(t;x,x)\sqrt{\det{g_u}}dx \\=& \frac{1}{4\pi t}\int_{\varphi^{-1}(R)}\psi d\Vol_{u} -\frac{1}{8\sqrt{\pi t}}\int_{\varphi^{-1}(B)}\psi\Big(1-\frac{2}{3}\sqrt{\frac{t}{\pi}}k_{u}\Big) d\ell_{u}\\ &
+\frac{1}{8\pi}\int_{\varphi^{-1}(B)}\partial_{n_{u}}d\ell_{u} + O(\vare) + O(t^{1/2})+O(\vare^2t^{-1/2})+O(t^{-1}e^{-c_{20}\vare^2/t}).\end{align*}
Finally, combining this with~\eqref{eq:levisumlocality} yields the desired result.
\end{proof}
\begin{lemma}\label{lemma:corner}Let $(M,g_0,(p_j),(\alpha_j))$, $\sigma,$ $\psi$, and $g_u$ be as in Theorem~\ref{thm:heattrace}. With the notation of Definition \ref{def:cpd}, fix $j\in\{1,...,n\}$ and suppose that $\vare_0>0$ is such that $$(\partial V_j\setminus\gamma)\cap B_{2\vare_0}=\varnothing.$$
There exists, for all $\vare\in(0,\vare_0)$, straight wedges $W_{\vare}^+$ and $W_{\vare}^-$ of opening angles $\alpha_\vare^+\pi$ and $\alpha_\vare^-\pi$ respectively such that $\alpha^\pm_\vare=\alpha_j+O(\vare)$ and
$$W^-_\vare\cap B_{2\vare}\subset V_{j}\cap B_{2\vare}\subset W^+_\vare\cap B_{2\vare}.$$ 
There also exists constants $t_1$ and $c_1,c_2,c_3,c_4$ such that, for all $\sqrt{t}<\vare<\vare_0$, $t<t_0$, and $u\in I$, where $I$ is a fixed compact interval, one has
\begin{align}\begin{split}&\bigg|\int_{\varphi^{-1}(\Omega)}H_u(t;p,p)d\Vol_{g_u}-\int_{\Omega_-}H_{(W^-_{\vare},dx)}(te^{-2\sigma_{j,u}(0)};x,x)dx\Big|\\\leq &\bigg|\int_{\Omega^+}H_{(W^+_\vare,dx)}(te^{-2\sigma_{j,u}(0)};x,x)dx-\int_{\Omega_-}H_{(W^-_\vare,dx)}(te^{-2\sigma_{j,u}(0)};x,x)dx\bigg|+E(\vare,t)\end{split}\label{eq:cornerbound}\end{align}
for all choices $\Omega\subset V_j\cap B_{3\vare/2}$ and $\Omega^-\subset \Omega\subset \Omega^+\subset B_{3\vare/2}\cap W^+_\vare$,
where $\sigma_{j,u}=\sigma_j+u\sigma\circ\varphi_j^{-1}$ and $E(\vare,t)$ is the form
$$E(\vare,t)=\frac{c_1\vare^2}{t}e^{-c_2\vare^2/t}+c_3\vare^2t^{-1/2}+c_4\vare^3t^{-1}.$$
\end{lemma}
\begin{proof}The regularity of $\partial M$ at $p_j$ guarantees the existence of $W^\pm_\vare$ with $\alpha_\vare^\pm = \alpha_j+O(\vare)$. This set-up is illustrated in Figure~\ref{fig:cornerwedge}. Further, the locality principle gives
\begin{align}\Big|H_u(t;\varphi^{-1}_j(x),\varphi^{-1}_j(x))-H_{(V_j,e^{2\sigma_{j,u}(x)}dx)}(t;x,x)\Big|&\leq \frac{c_5}{t}e^{-c_6\vare^2/t},\label{eq:localitycorner}
\end{align}
for sufficiently small $t$ and $x\in B_{3\vare/2}\cap V_j$. By Duhamel's principle we have, for $x\in V_j$ and sufficiently small $t$,
\begin{align}\begin{split}&\Big|H_{(V_j,e^{2\sigma_{j,u}}dx)}(t;x,x)-H_{(V_j,dx)}(te^{-2\sigma_{j,u}(x)};x,x)e^{-2\sigma_{j,u}(x)}\Big|\\
=& \bigg|\int_0^t\int_{V_j}H_{(V_j,e^{2\sigma_{j,u}}dx)}(s;x,y)(e^{2\sigma_{j,u}(x)}-e^{2\sigma_{j,u}(y)})\partial_t H_{(V_j,dx)}((t-s)e^{-2\sigma_{j,u}(x)};y,x)e^{-2\sigma_{j,u}(x)}dyds\bigg|\\
\leq &c_7t^{-1/2}\end{split}\label{eq:duhamelcorner}\end{align}
where the inequality follows from $(e^{2\sigma_{j,u}(x)}-e^{2\sigma_{j,u}(y)})=O(\dist_{M_0}(x,y))$,~\eqref{eq:gaussian}, and~\eqref{eq:Davies}. Using~\eqref{eq:gaussian} and~\eqref{eq:Davies} again we find
\begin{equation}\Big|H_{(V_j,dz)}(te^{-2\sigma_j(x)};x,x)-H_{(V_j,dz)}(te^{-2\sigma_j(0)};x,x)\Big|\leq c_8\frac{\vare}{t},\label{eq:derivativeboundcorner}\end{equation} 
for all $x\in B_{3\vare/2}$.
By domain monotonicity and the locality principle 
$$H_{(W^-_\vare,dx)}(t;x,x)-\frac{c_9}{t}e^{-c_{10}\vare^2/t}\leq H_{(V_j,dx)}(t;x,x)\leq H_{(W^+_\vare,dx)}(t;x,x)+\frac{c_9}{t}e^{-c_{10}\vare^2/t}$$
for sufficiently small $t$ and $x\in B_{3\vare/2}$. Hence,
\begin{align*}&\bigg|\int_{\Omega}  H_{(V_j,dx)}(t;x,x)dx-\int_{\Omega^-}H_{(W^-_\vare,dx)}(t;x,x)dx \bigg| \\ &\leq \bigg|\int_{\Omega^+}H_{(W^+_\vare,dx)}(t;x,x)dx-\int_{\Omega^-}H_{(W^-_\vare,dx)}(t;x,x)dx\bigg|+\frac{c_{11}\vare^2}{t}e^{-c_{10}\vare^2/t}.
\end{align*}
Combining this with~(\ref{eq:localitycorner}-\ref{eq:derivativeboundcorner}) finishes the proof.
\begin{figure}
\includegraphics[width=\linewidth]{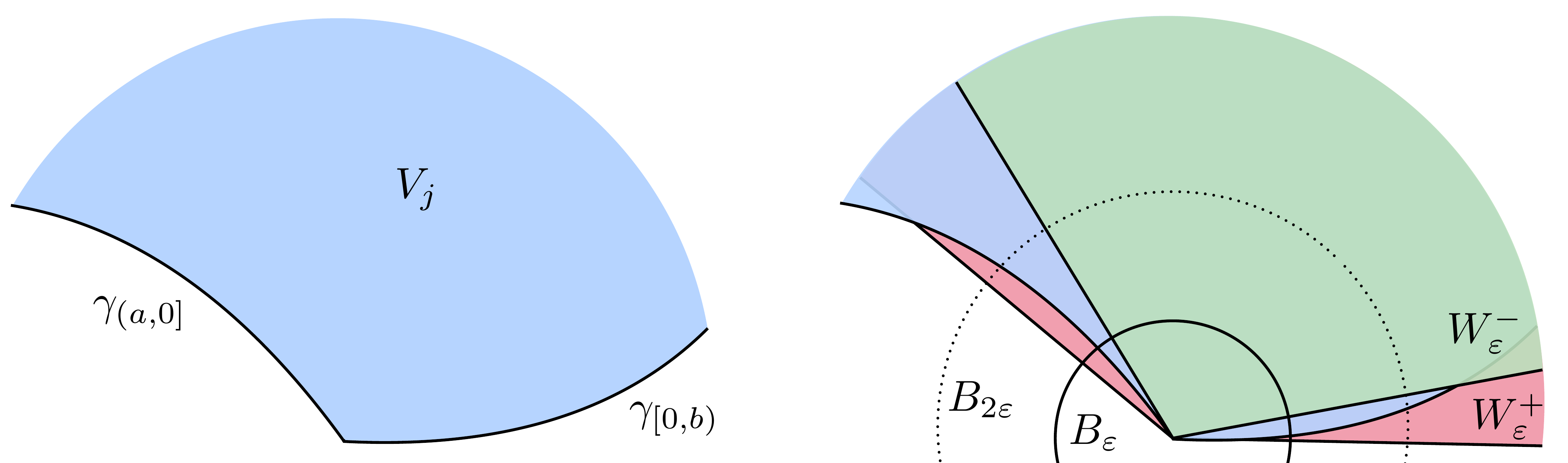}
\caption{Illustration of the set-up in Lemma~\ref{lemma:corner}\label{fig:cornerwedge}.}
\end{figure}
\end{proof}
\begin{proof}[Proof of Theorem~\ref{thm:heattrace}]
We will use Lemmas~\ref{lemma:interior}, \ref{lemma:boundary}, and \ref{lemma:corner} to estimate the integral of heat kernel along the diagonal locally. To this end, we construct a suitable partitioning of $M$, see Figure~\ref{fig:heattracesetup}. We will let the partitioning depend on a scale $\vare>0$, and the scale will depend on $t$. That is to say, we will estimate $\Tr(\psi H_{(M,g)}(t;\cdot,\cdot))$ using the partitioning of scale $\vare=\vare(t)$. Loosely, the partitioning will be constructed so that the locality principle can be used within distances $\vare$, which will give an error of the type $O(t^{-1}e^{-c_1\vare^2(t)/t})$. We therefore require that $\vare(t)=o(\sqrt{t})$.\par
For each corner $p_j$ and sufficiently small $\vare>0$, let $B_{j,\vare}=\varphi_j^{-1}(B_\vare\cap V_j)$ using the notation of Lemma~\ref{lemma:corner}. Next, observe that $\partial M\setminus\{p_1,...,p_n\}$ can be covered by rectangular coordinates as in Lemma~\ref{lemma:boundary}. For each corner $p_j$, we may also cover $\partial M$, locally at $p_j$, by the neighborhoods $V_{j,1}$ and $V_{j,2}$ (which are described at the end of Section \ref{section:cpd} and illustrated in Figure \ref{fig:Vjk}). Strictly speaking $V_{j,1}$ and $V_{j,2}$ are not subsets of $M$, but by using $V_{j,1}$ and $V_{j,2}$ we effectively obtain smooth rectangular coordinates close to $p_j$ with respect to the two boundary arcs meeting there. We may assume that the rectangular coordinates are compatible with each other, in the sense that the transition maps between different rectangular coordinates preserve rectangles (one way to achieve this is to further impose that the coordinates are boundary normal coordinates). In order to simplify estimates close to each corner we impose that the rectangular coordinates in $V_{j,k}$, $j=1,...,n$ and $k=1,2$ are boundary normal coordinates with respect to the Euclidean metric on $V_{j,k}$. Since $\partial M$ is compact we can restrict our attention to finitely many such rectangular coordinates. This gives a partitioning of $\partial M$ into finitely many sub-arcs $\eta_k$, $k=1,...,N$, where each sub-arc is compactly contained in a boundary rectangle. Let $c_2\in(0,1)$. We denote the ``rectangle of height $c_2\vare$ above $\eta_k$'' by $R_{k,\vare}$ (see Figure \ref{fig:heattracesetup}). The union $\cup_k R_{k,\vare}$ covers $\partial M$ and if $\vare$ is sufficiently small $R_{k,\vare}\cap R_{\ell,\vare} =\varnothing$ if $k\neq\ell$ unless $\eta_{k}$ and $\eta_{\ell}$ are adjacent to the same corner. However, by tuning the constant $c_2$ it can be achieved that, if $\eta_{k}$ and $\eta_{\ell}$ are adjacent to $p_j$, then
$$R_{k,\vare}\cap R_{\ell,\vare} \subset B_{j,\vare},$$
for sufficiently small $\vare>0$ ($c_2$ is controlled by $\alpha_{\min}=\min\{\alpha_1,...,\alpha_n,\pi\}$ and is roughly $c_2\approx\sin(\alpha_{\min}/2)$). For every small $\vare>0$ we partition $M$ into
$$B_{j,\vare},\ j=1,...,n,\quad R_{k,\vare}\setminus \cup_j B_{j,\vare},\ k=1,...,N,\quad M_\vare:=M\setminus(\cup_j B_{j,\vare}\cup_k R_{k,\vare}).$$For $u$ in a compact interval $I$, there exists a constant $c_3$ such that $\dist_{u}(x,\partial M)\geq c_3\vare$ for all $u\in I$ and $x\in M_\vare$.
\begin{figure}
\centering
\includegraphics[scale=0.25]{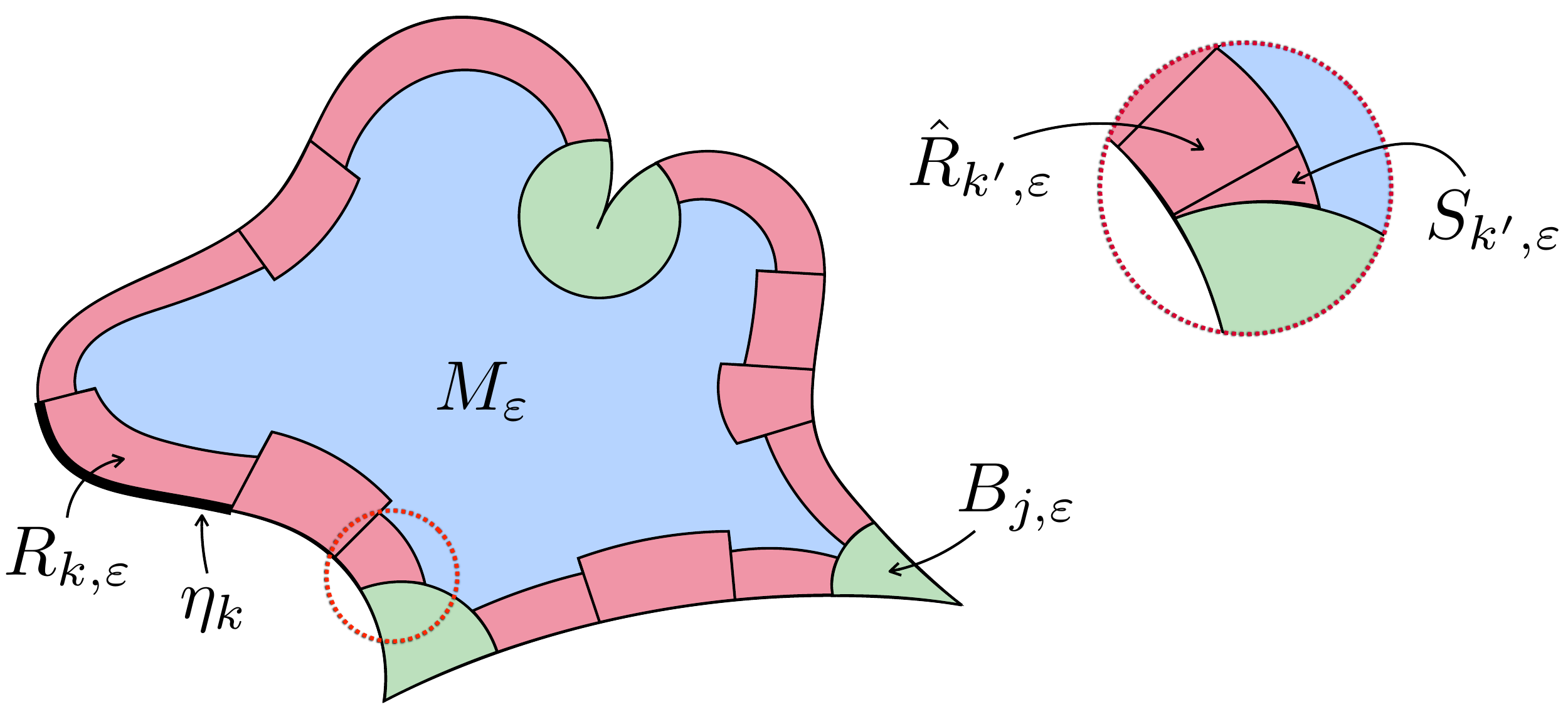}
\caption{Partitioning of $M$.\label{fig:heattracesetup}}
\end{figure}
Hence, Lemma~\ref{lemma:interior} shows that 
\begin{equation}\int_{M_\vare}\psi H_{u}(t;p,p)d\Vol_{u}=\frac{1}{4\pi t}\int_{M_\vare}\psi\Big(1+\frac{t}{3}K_{u}\Big)d\Vol_{u} +O(t)+O(t^{-1}e^{-c_4 \vare^2/t}),\label{eq:interior}\end{equation}
as $t\to 0+$.
Next, fix a $k\in\{1,...,N\}$ and consider $R_{k,\vare}\setminus \cup_j B_{j,\vare}$. If $\eta_k$ is adjacent to $p_j$, then $R_{k,\vare}\setminus B_{j,\vare}$ is not a rectangle in the local coordinates. Let $\hat R_{k,\vare}$ be the widest rectangle of height $c_2\vare$, such that $\hat R_{k,\vare}\subset R_{k,\vare}\setminus B_{j,\vare}$ and denote by 
$$S_{\vare,k}=R_{k,\vare}\setminus (B_{j,\vare}\cup \hat R_{k,\vare}),$$
see again Figure~\ref{fig:heattracesetup}.
If $\eta_k$ is not adjacent to a corner, then $R_{k,\vare}\setminus (\cup_j B_{j,\vare})=R_{k,\vare}$ so we set $\hat R_{k,\vare}=R_{k,\vare}$ and $S_{k,\vare}=\varnothing$.
Lemma~\ref{lemma:boundary} gives
\begin{align}\begin{split}&\int_{\cup_k \hat R_{k,\vare}} \psi H_u(t;p,p)d\Vol_{u}\\
&=\frac{1}{4\pi t}\int_{\cup_k \hat R_{k,\vare}}\psi\Big(1+\frac{t}{3}K_{u}\Big)d\Vol_{u}-\frac{1}{8\sqrt{\pi t}}\int_{\partial M\cap (\cup_k \hat R_{k,\vare})}\psi\Big(1-\frac{2}{3}\sqrt{\frac{t}{\pi}}k_{u}\Big)d\ell_{u}\\ &\quad +\frac{1}{8\pi}\int_{\partial M\cap (\cup_k \hat R_{k,\vare})}\partial_{n_{u}}\psi d\ell_{u}+O(t^{1/2})+O(\vare) +O(\vare^2 t^{-1/2})+O(t^{-1}e^{-c_5\vare^2/t}).\end{split}\label{eq:boundary}\end{align}
Here, the locality principle is used for $\hat R_{k,\vare}$ adjacent to $p_j$ to justify the usage of rectangular coordinates with respect to $V_{j,1}$ (or $V_{j,2}$), rather than $V_{j}$ ($V_{j}$ and $V_{j,1}$ coincide within a $c_6\vare$-neighborhood of $\hat R_{k,\vare}$).
\par
It remains to handle $B_{j,\vare}$ and $S_{k,\vare}$, which are contained in $\cup_j B_{j,3\vare/2}$ for small $\vare>0$. Hence, we may estimate the diagonal of the heat kernel using  Lemma~\ref{lemma:corner}. To do so, fix $j\in\{1,...,n\}$ and let $k_1$, $k_2$ be such that $\eta_{k_1}$ and $\eta_{k_2}$ are adjacent to $p_j$. After a re-labeling, $k_1=1$ and $k_2=2$. We now switch to $V_j$ coordinates and use the notation of Lemma~\ref{lemma:corner}. Consider $\varphi_j(S_{1,\vare})$ and let 
$S_{1,\vare}^-\subset \varphi_j(S_{1,\vare})\subset S_{1,\vare}^+$
be as in Figure~\ref{fig:sliver}. For a shape $S(\vare,h,b)$, with $b\geq 0,$ and $h\in(0,\vare)$, as in Figure~\ref{fig:sliver}, one can compute
\begin{align*}\int_{S(\vare,h,b)}&H_{(\{x_1>0\},dx)}(t;x,x)dx\\ &= \frac{b}{4\pi t}\int_0^h(1-e^{-x_2^2/t})dx_2 -\frac{1}{4\pi t}\int_0^h(1-e^{-x_2^2/t})(\vare-\sqrt{\vare^2-h^2})dx_2\\
&=\frac{hb}{4\pi t}-\frac{b+\vare}{8\sqrt{\pi t}}+\frac{\vare^2}{4\pi t}\int_0^{h/\vare}(1-e^{-(u\vare)^2/t})\sqrt{1-u^2}du+O((b+\vare)t^{-1/2}e^{-c_6h^2/t}).\end{align*}
Moreover, by~\cite[Theorem 2 and Corollary 3]{VDBS},
\begin{figure}
\centering
\includegraphics[width=0.8\linewidth]{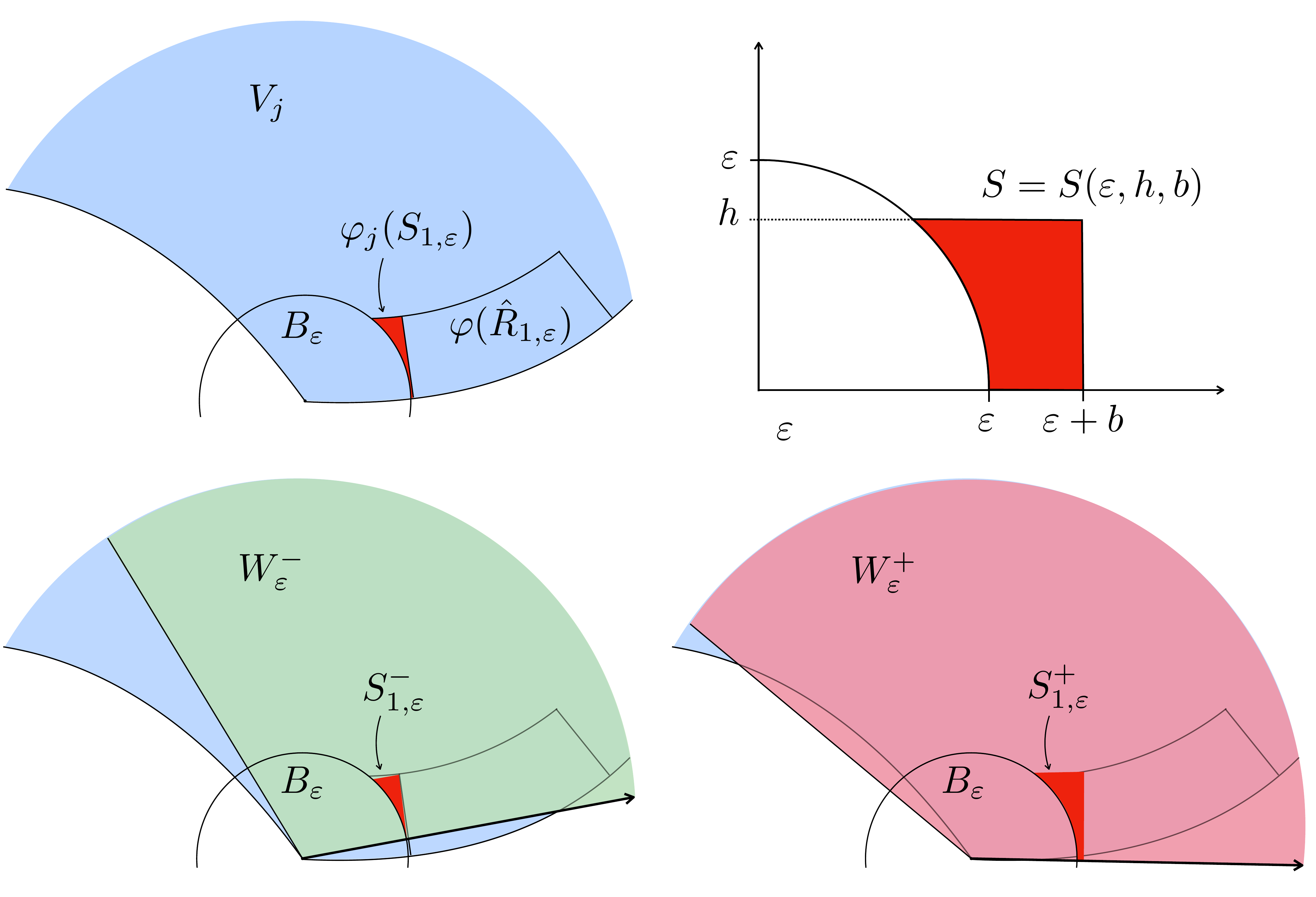}
\caption{On the top left we illustrate the set $\varphi_j(S_{1,\vare})$ and on the top right we illustrate the set $S=S(\vare,b,h)$ relative a ball of radius $\vare$ and the $x_1$-axis. In the bottom figures  we illustrate the two sets $S_{1,\vare}^+\subset W_{\vare}^+$ and $S_{1,\vare}^-\subset W_\vare^-$, which are of the type $S=S(\vare,h,b)$ relative the ball $B_\vare$ and the lower prong of $W_\vare^+$ and $W_\vare^-$ (indicated with an arrow) respectively.\label{fig:sliver}} 
\end{figure}
we have for a wedge $W=\{(r,\theta)\in M_0:\theta\in(0,\alpha\pi)\}$, $\alpha\in(0,2]$, that
\begin{equation}\int_{W \cap B_\vare}H_{(W,dx)}(t;x,x)dx = \frac{\alpha\pi\vare^2}{8\pi t}-\frac{\vare^2}{2\pi t}\int_0^1 e^{-(\vare u)^2/t}\sqrt{1-u^2}du +\frac{1-\alpha^2}{24\alpha}+A(t),\label{eq:VDBSwedge}\end{equation}
where 
$$|A(t)|\leq\begin{cases}\frac{\alpha}{8}e^{-\vare^2/t}& \alpha\in (1/2,2]\\ \frac{3}{64\alpha}e^{-(\vare\sin(\alpha\pi))^2/t}& \alpha\in(0,1/2].\end{cases}$$
Upon examining the proof, one is easily convinced that~\eqref{eq:VDBSwedge} holds also for $\alpha\in(2,\infty)$ with
\begin{equation}|A(t)|\leq \frac{\alpha}{2} e^{-\vare^2/t}.\label{eq:wedgebnd}
\end{equation}
This claim is shown in Appendix~\ref{section:appendix}. Construct $S_{2,\vare}^\pm$ analogously to $S_{1,\vare}^\pm$ and 
let $(h_{\ell,\vare}^+,b_{\ell,\vare}^+)$ and $(h_{\ell,\vare}^-,b_{\ell,\vare}^-)$ be the shape parameters corresponding to $S_{\ell,\vare}^+$ and $S_{\ell,\vare}^-$, $\ell=1,2$. Using the locality principle
and 
$$\frac{\vare^2}{4\pi t}\int_{h/\vare}^1e^{-(u\vare)^2/t}\sqrt{1-u^2}du = O(\vare^2t^{-1}e^{-c_{7}h^2/t}),$$
we obtain
\begin{align}\begin{split}&\int_{(W^\pm_\vare\cap B_{\vare})\cup S^\pm_{1,\vare}\cup S^\pm_{2,\vare}}H_{W_\vare^\pm}(t;x,x)dx\\ = &\frac{\alpha_\vare^\pm\pi\vare^2}{8\pi t}+\frac{\vare^2}{4\pi t}\int_0^{h^\pm_1/\vare}\sqrt{1-u^2}du+\frac{\vare^2}{4\pi t}\int_0^{h^\pm_2/\vare}\sqrt{1-u^2}du +\frac{h_{1,\vare}^\pm b_{1,\vare}^\pm+h_{2,\vare}^\pm b_{2,\vare}^\pm}{4\pi t}\\ & -\frac{b_{1,\vare}^\pm+b_{2,\vare}^\pm+2\vare}{8\sqrt{\pi t}}+\frac{1-(\alpha_\vare^\pm)^2}{24\alpha_\vare^\pm}+O(\vare^2t^{-1}(e^{-c_{10}(h_{1,\vare}^\pm)^2/t}+e^{-c_{7}(h_{2,\vare}^\pm)^2/t}))\\ &+O((b_{1,\vare}^\pm+\vare)t^{-1/2}e^{-c_6(h_{1,\vare}^\pm)^2/t})+O((b_{2,\vare}^\pm+\vare)t^{-1/2}e^{-c_6(h_{2,\vare}^\pm)^2/t})+O(e^{-c_{8}\vare^2/t}).\end{split}\label{eq:corner1}\end{align}
It is possible to choose $W^\pm_\vare$ (according to Lemma~\ref{lemma:corner}), and $S_{\ell,\vare}^\pm$, $\ell=1,2$ in such a way that, for $\ell=1,2$,
$$b_{\ell,\vare}^-=0,\quad h_{\ell,\vare}^-=c\vare+O(\vare^2),\quad b_{\ell,\vare}^+=O(\vare^2),\quad h_{\ell,\vare}^+=c\vare+O(\vare^2),$$
as $\vare\to 0+$. 
Then,~\eqref{eq:corner1} becomes
\begin{align}\begin{split}\int _{(W^\pm_\vare\cap B_{\vare})\cup S^\pm_{1,\vare}\cup S^\pm_{2,\vare}}H_{W_\vare^\pm}(t;x,x)dx = &\frac{\alpha_j\pi\vare^2}{8\pi t}+\frac{\vare^2}{2\pi t}\int_0^{c}\sqrt{1-u^2}du -\frac{2\vare}{8\sqrt{\pi t}}+\frac{1-\alpha_j^2}{24\alpha_j}\\&+O(\vare^3t^{-1})+O(\vare^2t^{-1/2})+O(t^{-1}e^{-c_{9}\vare^2/t}).\end{split}\label{eq:corner2}\end{align}
Since,
$$(W^-_\vare\cap B_\vare) \cup S^-_{1,\vare}\cup S^-_{2,\vare} \subset \varphi_j(B_{j,\vare}\cup S_{1,\vare}\cup S_{2,\vare})\subset (W^+_\vare\cap B_\vare) \cup S^+_{1,\vare}\cup S^+_{2,\vare}\subset B_{3\vare/2},$$
Lemma~\ref{lemma:corner} yields
\begin{align*}&\Big|\int_{B_{j,\vare}\cup S_{1,\vare}\cup S_{2,\vare}} H_{u}(t;p,p)d\Vol_{u} - \int_{(W^-_\vare\cap B_\vare) \cup S^-_{1,\vare}\cup S^-_{2,\vare}} H_{(W^-_\vare,dx)}(te^{-2\sigma_{j,u}(0)};x,x)dx\Big|\\
&=O(\vare^3t^{-1})+O(\vare^2t^{-1/2})+O(t^{-1}e^{-c_{10}\vare^2/t}).\end{align*}
Hence,
\begin{align*}\int_{B_{j,\vare}\cup S_{1,\vare}\cup S_{2,\vare}} H_{u}(t;p,p)d\Vol_{u} =& \frac{\alpha_j\pi\vare^2e^{2\sigma_{j,u}(0)}}{8\pi t}+\frac{\vare^2 e^{2\sigma_{j,u}(0)}}{2\pi t}\int_0^c\sqrt{1-u^2}du-\frac{2\vare e^{\sigma_{j,u}(0)}}{8\sqrt{\pi t}}\\&+\frac{1-\alpha^2_j}{24\alpha_j}+O(\vare^3t^{-1})+O(\vare^2t^{-1/2})+O(t^{-1}e^{-c_{10}\vare^2/t}).\end{align*}
The estimates of $b_{\ell,\vare}^\pm$ and $h_{\ell,\vare}^\pm$ as well as the regularity of the boundary at the corner shows that
\begin{align*}
&\Vol_{u}(B_{j,\vare})=\frac{\alpha_j\pi\vare^2e^{2\sigma_{j,u}(0)}}{2}+O(\vare^3)
\\
&\Vol_{u}(S_{k,\vare})=\vare^2e^{2\sigma_{j,u}(0)}\int_0^c\sqrt{1-u^2}du+O(\vare^3),\quad k=1,2.
\\
&\ell_{u}(\partial M\cap (B_{j,\vare}\cup S_{1,\vare}\cup S_{2,\vare}))=2\vare e^{\sigma_{j,u}}(0)+O(\vare^2).
\end{align*}
We deduce,
\begin{align*}\int_{B_{j,\vare}\cup S_{1,\vare}\cup S_{2,\vare}} H_{u}(t;p,p)d\Vol_{u} =& \frac{\Vol_{u}(B_{j,\vare}\cup S_{1,\vare}\cup S_{2,\vare})}{4\pi t}-\frac{\ell_{u}(\partial M\cap (B_{j,\vare}\cup S_{1,\vare}\cup S_{2,\vare}))}{8\sqrt{\pi t}} \\ &+\frac{1-\alpha^2_j}{24\alpha_j}
 +O(\vare^3t^{-1})+O(\vare^2t^{-1/2})+O(t^{-1}e^{-c_{10}\vare^2/t}).\end{align*}
Since,
\begin{align*}&\int_{B_{j,\vare}\cup S_{1,\vare}\cup S_{2,\vare}} (\psi(p)-\psi(p_j)) H_{(M,g_u)}(t;p,p)d\Vol_{u}=O(\vare)\int_{B_{j,\vare}\cup S_{1,\vare}\cup S_{2,\vare}} H_{(M,g_u)}(t;p,p)d\Vol_{u}\end{align*}
and
\begin{align*}\int_{B_{j,\vare}\cup S_{1,\vare}\cup S_{2,\vare}} \psi(p)d\Vol_{g_u} &= \psi(p_j)\Vol_{u}(B_{j,\vare}\cup S_{1,\vare}\cup S_{2,\vare})+O(\vare^3),\\
\int_{\partial M\cap (B_{j,\vare}\cup S_{1,\vare}\cup S_{2,\vare})} \psi(p)d\ell_{g_u} &= \psi(p_j)\ell_{u}(\partial M\cap (B_{j,\vare}\cup S_{1,\vare}\cup S_{2,\vare}))+O(\vare^2),\\
\int_{B_{j,\vare}\cup S_{1,\vare}\cup S_{2,\vare}} \psi(p)K_{u} d\Vol_{u} &= O(\vare^2),\\
\int_{\partial M\cap (B_{j,\vare}\cup S_{1,\vare}\cup S_{2,\vare})} \psi(p)k_{u}d\ell_{u} &= O(\vare),\\
\int_{\partial M\cap (B_{j,\vare}\cup S_{1,\vare}\cup S_{2,\vare})} \partial_{n_{u}}\psi(p)d\ell_{u} &= O(\vare),\\
\end{align*}
we obtain,
\begin{align}\begin{split}&\int_{B_{j,\vare}\cup S_{1,\vare}\cup S_{2,\vare}} \psi(p)H_{u}(t;p,p)d\Vol_{u}\\
&= \frac{1}{4\pi t}\int_{B_{j,\vare}\cup S_{1,\vare}\cup S_{2,\vare}} \psi(p)\Big(1+\frac{t}{3}K_{u}\Big)d\Vol_{u}+\frac{1}{8\pi}\int_{\partial M\cap (B_{j,\vare}\cup S_{1,\vare}\cup S_{2,\vare})}\partial_{n_{u}}\psi d\ell_{u}\\ 
&\quad -\frac{1}{8\sqrt{\pi t}}\int_{\partial M\cap (B_{j,\vare}\cup S_{1,\vare}\cup S_{2,\vare})} \psi(p)\Big(1-\frac{2}{3}\sqrt{\frac{t}{\pi}}k_{u}\Big)d\ell_{u}+\frac{1-\alpha_j^2}{24\alpha_j}\psi(p_j)\\
&\quad +O(\vare^3t^{-1})+O(\vare^2t^{-1/2})+O(\vare)+O(t^{-1}e^{-c_{10}\vare^2/t}).\end{split}\label{eq:corner}\end{align}
Together~(\ref{eq:interior}, \ref{eq:boundary}, \ref{eq:corner}) imply
\begin{align*}\int_M\psi H_u(t;p,p)d\Vol_u =& \frac{1}{4\pi t}\int_M\psi\Big(1+\frac{t}{3}K_u\Big)d\Vol_{u}-\frac{1}{8\sqrt{\pi t}}\int_{\partial M}\psi\Big(1-\frac{2}{3}\sqrt{\frac{t}{\pi}}k_u\Big)d\ell_u\\ & +\frac{1}{8\pi}\int_{\partial M}\partial_{n_u}\psi d\ell_u +\sum_{j=1}^n\frac{1-\alpha_j^2}{24\alpha_j}\psi(p_j)+O(t^{1/2})+O(\vare)\\ &+O(\vare^2 t^{-1/2})+O(\vare^3 t^{-1})+O(t^{-1}e^{-c_{11}\vare^2/t}).\end{align*}
By setting $\vare=t^{(q+1)/3}$ for $q\in(0,1/2)$, all errors on the right hand side are $O(t^q)$. This finishes the proof.
\end{proof}
\section{Polyakov-Alvarez type anomaly formula}\label{section:polalv}
This section is devoted to proving Theorem~\ref{thm:polyakov}. Let $(M,g_0,(p_j),(\alpha_j))$ be a curvilinear polygonal domain, $\sigma\in C^\infty(M,g_0,(p_j),(\alpha_j)),$ and $g_u=e^{2u\sigma}g$. To simplify notation, we write 
$$\Delta_u=\Delta_{(M,g_u)},\quad \zeta_u=\zeta_{(M,g_u)},\quad H_{(M,g_u)}=H_u,\quad L^2(M,\Vol_{g_u})=L^2_u,$$
$$d\Vol_{u}=d\Vol_{g_u},\quad d\ell_{u}=d\ell_{g_u},\quad K_u=K_{g_u},\quad k_{u}=k_{g_u},\quad \partial_{n_u}=\partial_{n_{g_u}},$$
For $\psi\in C^\infty(M,g_0,(p_j),(\alpha_j))$ we write
$$a_{-1}(u,\psi)=\frac{1}{4\pi}\int_M \psi d\Vol_{u},\quad a_{-1/2}(u,\psi)=-\frac{1}{8\sqrt{\pi}}\int_{\partial M}\psi d\ell_{u},$$
$$a_0(u,\psi)=\frac{1}{12\pi}\int_{M}\psi K_{u}d\Vol_{u}+\frac{1}{12\pi}\int_{\partial M}\psi k_{u}d\ell_{u}+\frac{1}{8\pi}\int_{\partial M}\partial_{n_{u}}\psi d\ell_{u}+\frac{1}{24}\sum_{j=1}^n\frac{1-\alpha_j^2}{\alpha_j}\sigma(p_j),$$
and $a_{-1}(u)=a_{-1}(u,1),$ $a_{-1/2}(u)=a_{-1/2}(u,1),$ and $a_{0}(u)=a_{0}(u,1).$
To prove Theorem~\ref{thm:polyakov} we need the following lemma.
\begin{lemma}\label{lemma:difftrace}Let $(M,g_0,(p_j),(\alpha_j))$ be a curvilinear polygonal domain, $\sigma\in C^\infty(M,g_0,(p_j),(\alpha_j))$, and $g_u=e^{2u\sigma}g_0$. For every $\vare>0$ and $u\in\R$ 
$$\partial_u \int_\vare^\infty \frac{1}{t}\Tr(e^{-t\Delta_u})dt = 2\Tr(\sigma e^{-\vare\Delta_{u}}).$$
\end{lemma}
\begin{proof}
We first show that $H_u(t;x,y)e^{2u\sigma(y)}$ is differentiable in $u$ for every fixed $(t;x,y)$ in $(0,\infty)\times M^\circ\times M^\circ$. By Duhamel's principle
\begin{align*}&H_u(t;x,y)e^{2u\sigma(y)}-H_0(te^{-2u\sigma(y)};x,y)\\ &=\int_0^t\int_M H_u(s;x,z)e^{2u\sigma(z)}(e^{-2u\sigma(y)}-e^{-2u\sigma(z)})\Delta_{0,z}H_0((t-s)e^{-2u\sigma(y)};z,y)d\Vol_{g_0}(z)ds.
\end{align*}
By~\eqref{eq:gaussian},~\eqref{eq:Davies}, and smoothness of $\sigma$,
\begin{align}\begin{split}H_u(t,x,y)&\leq \frac{c_1}{t}e^{-c_2 d_u(x,y)^2/t},\\
|\Delta_{0,x}H_u(te^{-2u\sigma(y)},x,y)| &= |e^{2u\sigma(y)}\partial_t H_u(te^{-2u\sigma(y)},x,y)|\leq \frac{c_3}{t^2}e^{-c_4 d_u(x,y)^2/t},\\
|e^{-2u\sigma(y)}-e^{-2u\sigma(z)}| &\leq c_5|u|d_{u}(x,y),\end{split}\label{eq:duhamelbounds}\end{align}
for sufficiently small $t$ and $u$ in a compact interval. By applying Duhamel's principle again and using~\eqref{eq:duhamelbounds} we find
\begin{align*}&H_u(t;x,y)e^{2u\sigma(y)}-H_0(te^{-2u\sigma(y)};x,y)\\&=\int_0^t\int_M H_0(se^{-2\sigma(z)};x,z)(e^{-2u\sigma(y)}-e^{-2u\sigma(z)})\Delta_{0,z}H_0((t-s)e^{-2u\sigma(y)};z,y)d\Vol_{g_0}(z)ds+O(u^2),\end{align*}
as $u\to 0$ and $t\to 0+$. Hence,
\begin{align*}&\partial_u H_u(t;x,y)e^{2u\sigma(y)}|_{u=0}\\&=-2\sigma(y)t\partial_t H_0(t;x,y)+\int_0^t\int_M H_0(s;x,z)(2\sigma(y)-2\sigma(z))\partial_t H_0((t-s);z,y)d\Vol_{g_0}(z)ds.\end{align*}
There is nothing special about $u=0$, so we in fact have
\begin{align*}&\partial_u H_u(t;x,y)e^{2(u-u_{0})\sigma(y)}|_{u=u_0}\\ &=-2\sigma(y)t\partial_t H_{u_0}(t;x,y)+\int_0^t\int_M H_{u_0}(s;x,z)(2\sigma(y)-2\sigma(z))\partial_t H_{u_0}((t-s);z,y) d\Vol_{g_{u_0}}(z)ds.\end{align*}
By integrating along the diagonal (and noting that~\eqref{eq:duhamelbounds} gives a locally uniform and integrable upper bound on the derivative) we find
$$\partial_u \Tr(e^{-t\Delta_u})=-2t\int_M \sigma(x)\partial_t H_u(t;x,x)d\Vol_{g_u}.$$
Finally, we obtain
$$\partial_u \int_\vare^\infty\frac{1}{t} \Tr(e^{-t\Delta_u})dt = 2\int_M \sigma(x) H_u(\vare;x,x)d\Vol_{g_u} = 2\Tr(\sigma e^{-\vare\Delta_u}).$$
\end{proof}
\begin{proof}[Proof of Theorem~\ref{thm:polyakov}]
Recall, from Section \ref{section:zdet}, that 
\begin{align*}\zeta_u(s)=&\frac{1}{\Gamma(s)}\bigg(\int_0^1t^{s-1}(\Tr(e^{-t\Delta_u})-a_{-1}(u)t^{-1}-a_{-1/2}(u)t^{-1/2}-a_0(u))dt\\ &+\int_1^\infty t^{s-1}\Tr(e^{-t\Delta_u})dt+\frac{a_{-1}(u)}{s-1}+\frac{a_{-1/2}}{s-1/2}\bigg)+\frac{1}{s\Gamma(s)}a_0(u).\end{align*}
Using that $1/\Gamma(s)=s+\pmb \gamma s^2 +O(s^3)$, where $\pmb \gamma$ denotes the Euler-Mascheroni  constant, we find
\begin{align*}\zeta_u'(0)=&\int_0^1t^{-1}(\Tr(e^{-t\Delta_u})-a_{-1}(u)t^{-1}-a_{-1/2}(u)t^{-1/2}-a_0(u))dt\\ &+\int_1^\infty t^{-1}\Tr(e^{-t\Delta_u})dt+\frac{a_{-1}(u)}{-1}+\frac{a_{-1/2}}{-1/2}+\pmb \gamma a_0(u)\\ =&\lim_{\vare\to 0+}\bigg(\int_\vare^\infty t^{-1}\Tr(e^{-t\Delta_u})dt-a_{-1}(u)\vare^{-1}-2a_{-1/2}(u)\vare^{-1/2}-a_0(u)\log \vare+\pmb \gamma a_0(u)\bigg).\end{align*}
Observe that
$$\partial_u a_{-1}(u)=\frac{1}{4\pi}\partial_u\int_M e^{2u\sigma}d\Vol_{0}=\frac{1}{4\pi}\int_M 2\sigma d\Vol_{u} = 2a_{-1}(u,\sigma),$$
$$\partial_u a_{-1/2}(u)=-\frac{1}{8\sqrt\pi}\partial_u\int_{\partial M} e^{u\sigma}d\ell_{0}=-\frac{1}{8\sqrt\pi}\int_{\partial M}\sigma d\ell_{u} = a_{-1/2}(u,\sigma),$$
while
\begin{align*}a_0(u)&=\frac{1}{12\pi}\int_M K_{u}d\Vol_{u} +\frac{1}{12\pi}\int_{\partial M}k_{u}d\ell_{u}+\frac{1}{24}\sum_{j=1}^n\frac{1-\alpha_j^2}{\alpha_j}\\
&=\frac{1}{12\pi}\int_M (K_{0}+u\Delta_0\sigma)d\Vol_{0} +\frac{1}{12\pi}\int_{\partial M}(k_{0}+u\partial_{n_u}\sigma)d\ell_{u}+\frac{1}{24}\sum_{j=1}^n\frac{1-\alpha_j^2}{\alpha_j}=a_0(0).\end{align*}
Define, for every $\vare>0$,
$$F_\vare(u):=\int_\vare^\infty t^{-1}\Tr(e^{-t\Delta_u})dt-a_{-1}(u)\vare^{-1}-2a_{-1/2}(u)\vare^{-1/2}-a_0(u)\log \vare+\pmb \gamma a_0(u).$$
By Lemma~\ref{lemma:difftrace} and the above
$$F'_\vare(u)=2\Tr(\sigma e^{-\vare\Delta_{u}})-2\frac{a_{-1}(u,\sigma)}{\vare}-2\frac{a_{-1/2}(u,\sigma)}{\vare^{1/2}}.$$
The short time asymptotic expansion from Theorem~\ref{thm:heattrace} implies that $F_\vare(u)\to \zeta'_u(0)$ and $F'_\vare(u)\to 2a_0(u,\sigma)$ locally uniformly in $u$ as $\vare\to 0+$. Hence,
$$-\partial_u\log\zdet\Delta_u =\partial_u\zeta_u'(s)=2a_0(u,\sigma),$$
which proves~\eqref{eq:polyakovdifferentiated} and furthermore,
$$\log\zdet\Delta_0-\log\zdet\Delta_1=\int_0^1 2 a_0(u,\sigma)du.$$
We have
$$\int_0^1\int_{M}\sigma K_u d\Vol_u =\int_0^1\int_M\sigma(K_0+u\Delta_0\sigma)d\Vol_{0}=\int_{M}\sigma K_0 d\Vol_{0} + \frac{1}{2}\int_{M}\sigma\Delta_0\sigma d\Vol_{0},$$
$$\int_0^1\int_{\partial M}\sigma k_u d\ell_u =\int_0^1\int_{\partial M}\sigma(k_0+u\partial_{n_0}\sigma)d\ell_{0}=\int_{\partial M}\sigma k_0 d\ell_{0} + \frac{1}{2}\int_{\partial M}\sigma\partial_{n_0}\sigma d\ell_{0},$$
$$\int_0^1\int_{\partial M}\partial_{n_u}\sigma d\ell_u = \int_{\partial M}\partial_{n_0}\sigma d\ell_0.$$
By Stokes' theorem 
$$\int_{M}\sigma\Delta_0 \sigma d\Vol_{g_0}+\int_{\partial M}\sigma_{n_0}\sigma d\ell_{g_0} = \int_{M}|\nabla_{g_0}\sigma|^2d\ell_{g_0}.$$
This yields~\eqref{eq:polyakovintegrated}. 
\end{proof}
\appendix
\section{Heat kernel on an infinite straight wedge with \texorpdfstring{$\alpha>2$}{}}\label{section:appendix}
In this section, we show that~\eqref{eq:wedgebnd} holds when $\alpha>2$. In~\cite{VDBS}, they derive~\eqref{eq:VDBSwedge} by considering the Green's function 
$$G_{W}(s;z,w)=\int_0^\infty e^{-st} H_{W}(t;z,w),$$
solving 
$$\begin{cases}sG_{W}+\Delta_z G_{W} = \delta_{w}(z),\\
G|_{\partial W}=0.\end{cases}$$
The Green's function can be expressed using a Kontorovich-Lebedev transform and radial coordinates $z=re^{i\theta}, w=\rho e^{i\phi}$ with $\theta,\phi\in(0,\alpha\pi)$
\begin{align*}G_{W}(s;z,w)=&\frac{1}{\pi^2}\int_0^\infty dx K_{ix}(r\sqrt s )K_{ix}(\rho\sqrt{s})\\ &\cdot\bigg[\cosh(\pi-|\theta-\phi|)x - \frac{\sinh \pi x}{\sinh \alpha\pi x}\cosh(\alpha-\theta-\phi)x+\frac{\sinh (1-\alpha)\pi x}{\sinh \alpha\pi x}\cosh(\theta-\phi)x \bigg].\end{align*}
Details can be found in~\cite[Appendix A]{NRS}. As $G_{W}$ is found in radial coordinates the expression for $G_W$ above is valid also for $\alpha>2$. The steps of the computation of~\eqref{eq:VDBSwedge} in~\cite{VDBS} are not dependent on the size of the angles, and they find that 
$$A(t)=-\mathcal L^{-1}\bigg\{\frac{\alpha}{\pi}\int_0^\infty dx \frac{\sinh(1-\alpha)\pi x}{\sinh \alpha\pi x}\int_\vare^\infty K^2_{ix}(r\sqrt{s})rdr\bigg\}.$$
They also show that, whenever $|\gamma_0|<\gamma_1$,
\begin{align*}\mathcal L^{-1}&\bigg\{\int_0^\infty dx \frac{\sinh\gamma_0 x}{\sinh \gamma_1 x}\int_\vare^\infty K^2_{ix}(r\sqrt{s})rdr\bigg\}\\ &=\frac{\pi}{4\gamma_1}\sin\frac{\pi\gamma_0}{\gamma_1}\int_0^\infty dq e^{-\vare^2(1+\cosh q)/(2t)}(1+\cosh q)^{-1}\Big(\cosh\frac{\pi q}{\gamma_1}+\cos\frac{\pi\gamma_0}{\gamma_1}\Big)^{-1}.\end{align*}
Applying this with $\gamma_0=(1-\alpha)\pi$, $\gamma_1=\alpha\pi$ for $\alpha\geq 2$ we may bound $A(t)$ by
\begin{align*}|A(t)|&=\frac{1}{4\pi}\bigg|\sin\frac{\pi}{\alpha}\int_0^\infty dq e^{-\vare^2(1+\cosh q)/(2t)}(1+\cosh q)^{-1}\Big(\cosh\frac{ q}{\alpha}-\cos\frac{\pi}{\alpha}\Big)^{-1}\bigg|\\
&\leq \frac{1}{8\alpha}e^{-\vare^2/t}\int_0^\infty \Big|\cosh\frac{q}{\alpha}-\cos\frac{\pi}{\alpha}\Big|^{-1}dq.\end{align*}
Using $\cosh\frac{q}{\alpha}\geq 1+\frac{q^2}{2\alpha^2}$ and $\cos\frac{\pi}{\alpha}\leq 1-\frac{\pi^2}{2\alpha^2}+\frac{\pi^4}{4!\alpha^4}$ yields
\begin{align*}\cosh\frac{ q}{\alpha}-\cos\frac{\pi}{\alpha}\geq \frac{1}{2\alpha^2}\Big(q^2+\frac{\pi^2}{2}\Big)\end{align*}
so that 
$$\int_{0}^\infty \Big(\cosh\frac{ q}{\alpha}-\cos\frac{\pi}{\alpha}\Big)^{-1}dq\leq \int_0^1 2\alpha^2\frac{2}{\pi^2}dq +\int_1^\infty  2\alpha^2 q^{-2}dq \leq 4\alpha^2.$$
We conclude, as desired, that
$$|A(t)|\leq \frac{\alpha}{2}e^{-\vare^2/t}.$$
\bibliographystyle{acm}
\bibliography{ref}
\end{document}